\documentclass[journal,twoside,web]{ieeecolor}
\usepackage{generic}
\usepackage{cite}
\usepackage{amsmath,amssymb,mathrsfs,amsfonts}
\usepackage{bbm}
\usepackage{algorithmic}
\usepackage{graphicx}
\usepackage{epstopdf}
\usepackage{subcaption}
\usepackage{textcomp}

\newtheorem{problem}{Problem}
\newtheorem{definition}{Definition}

\newtheorem{theorem}{Theorem}
\newtheorem{lemma}{Lemma}
\newtheorem{corollary}{Corollary}

\newtheorem{remark}{Remark}

\begin{document}

\title{Distributed Mean-Field Density Estimation for Large-Scale Systems}
\author{Tongjia Zheng$^{1}$, \IEEEmembership{Graduate Student Member, IEEE}, Qing Han$^{2}$ and Hai Lin$^{1}$, \IEEEmembership{Senior Member, IEEE}
\thanks{*This work was supported by the National Science Foundation under Grant No. IIS-1724070, CNS-1830335, IIS-2007949.}
\thanks{$^{1}$Tongia Zheng and Hai Lin are with the Department of Electrical Engineering, University of Notre Dame, Notre Dame, IN 46556, USA. {\tt\small  tzheng1@nd.edu, hlin1@nd.edu.}} 
\thanks{$^{2}$Qing Han is with the Department of Mathematics, University of Notre Dame, Notre Dame, IN 46556, USA. {\tt\small Qing.Han.7@nd.edu.}} 
}

\maketitle

\thispagestyle{empty}
\pagestyle{empty}

\begin{abstract}
This work studies how to estimate the mean-field density of large-scale systems in a distributed manner.
Such problems are motivated by the recent swarm control technique that uses mean-field approximations to represent the collective effect of the swarm, wherein the mean-field density (especially its gradient) is usually used in feedback control design.
In the first part, we formulate the density estimation problem as a filtering problem of the associated mean-field partial differential equation (PDE), for which we employ kernel density estimation (KDE) to construct noisy observations and use filtering theory of PDE systems to design an optimal (centralized) density filter.
It turns out that the covariance operator of observation noise depends on the unknown density.
Hence, we use approximations for the covariance operator to obtain a suboptimal density filter, and prove that both the density estimates and their gradient are convergent and remain close to the optimal one using the notion of input-to-state stability (ISS).
In the second part, we continue to study how to decentralize the density filter such that each agent can estimate the mean-field density based on only its own position and local information exchange with neighbors.
We prove that the local density filter is also convergent and remains close to the centralized one in the sense of ISS.
Simulation results suggest that the centralized suboptimal density filter is able to generate convergent density estimates, and the local density filter is able to converge and remain close to the centralized filter.
\end{abstract}


\section{Introduction}
Recent years have seen increased research activities in the study of large-scale agent systems, which treat the swarm as a continuum and use its mean-field density to represent the collective effect of all the agents \cite{lasry2007mean, ferrari2016distributed, yang2018mean, elamvazhuthi2019mean}.
More recently, the technique of mean-field feedback control has been proved to be promising for designing control laws for these mean-field models.
In this approach, the global density and its gradient of the swarm is fed back into the algorithms to form a closed-loop framework, which provides extra stability and robustness properties in the macroscopic level \cite{eren2017velocity, bandyopadhyay2017probabilistic, krishnan2018distributed, elamvazhuthi2018pde, zheng2020transporting}.
Considering the scalability issue of the control algorithms and the privacy issue of data sharing, it is desirable to implement the control strategy in a fully distributed manner using only local observations and information exchange.
This motivates the distributed mean-field density (and gradient) estimation problem.

Density estimation is a fundamental problem in statistics and has been studied using various methods, including parametric and nonparametric methods.
Parametric algorithms assume that the samples are drawn from a known parametric family of distributions and try to determine the parameters to maximize the likelihood using the samples \cite{silverman1986density}. 
Several distributed parametric techniques exist in the literature for estimating a distribution from a data set. 
For example, in \cite{nowak2003distributed, gu2008distributed}, the unknown density is represented by a mixture of Gaussians, and the parameters are determined by a combination of expectation maximization (EM) algorithms and consensus protocols.
Performance of such estimators rely on the validity of the assumed models, and therefore they are unsuitable for estimating an evolving density. 
In non-parametric approaches, the data are allowed to speak for themselves in determining the density estimate. 
\textit{Kernel density estimation (KDE)} \cite{silverman1986density} is probably the most popular algorithm in this category.
A distributed KDE algorithm is presented in \cite{hu2007distributed}, which uses an information sharing protocol to incrementally exchange kernel information between sensors until a complete and accurate approximation of the global KDE is achieved by each sensor.
All these distributed algorithms are proposed for estimating a static/stationary density.
To the best of our knowledge, distributed estimation for dynamic/non-stationary densities remains largely unexplored.

In the density estimation problem of large-scale systems, the agents' dynamics are usually available because they are built by task designers and follow the designed control laws.
This motivates the problem of how to take advantage of the available dynamics to improve the density estimation.
A candidate solution is to consider a filtering problem to estimate the distribution of all the agents' states.
There is a large body of literature for the filtering problem, ranging from the celebrated Kalman filters and their variants \cite{julier2004unscented} to the more general Bayesian filters and their Monte Carlo approaches, also known as particle filters \cite{chen2003bayesian}.
Motivated by the development of sensor networks, distributed implementation for these filters has also been extensively studied \cite{olfati2007distributed, bandyopadhyay2014distributed, battistelli2016stability}.
The general procedure of these distributed filtering algorithms is that each agent implements a local filter using its local observation, and exchanges its observation or/and estimates with nearby agents to gradually improve its estimate of the global distribution.
However, stability analysis and numerical implementation are difficult especially when the agents' dynamics are nonlinear and time-varying.
Also, it is unclear whether the obtained distribution is indeed the mean-field density of the agents.

In summary, considering the requirements of decentralization, convergence, and efficiency, existing methods are unsuitable for estimating the time-varying/non-stationary density of large-scale agent systems in a distributed manner. 
This motivates us to propose a distributed, dynamic, and scalable algorithm that can perform online and use only local observation and information exchange to obtain convergent density estimates. 
We tackle this problem by formulating the density estimation problem as a state estimation problem of the PDE system that governs the spatial and temporal evolution of the mean-field density and using KDE to construct noisy observations for this PDE system.
We show that the observation noise is approximately Gaussian but its covariance operator depends on the unknown density.
Hence, we use infinite-dimensional Kalman filters for PDE systems to design a (centralized) density filter and prove its stability and optimality when approximations are used for the unknown covariance operator.
Then, we decentralize the density filter such that each agent estimates the mean-field density based on only its own position and local information exchange with neighbors, and prove that the local density filter has comparable performance with centralized one in the sense of ISS.
Our contribution is summarized as follows: (i) we present a centralized density filter and prove that the density estimates and their gradient are both convergent; (ii) we present a distributed density filter such that each agent estimates the global density using only its own position and local communication; (iii) all the results hold even if the agents' dynamics are nonlinear and time-varying.

Some of the results have been reported in our previous work \cite{zheng2020pde, zheng2021distributed}.
The difference is clarified as follows.
In \cite{zheng2020pde}, we presented a density filter and proved its stability assuming that the filter is well-posed in some Hilbert space.
In this work, we will rigorously study the solution properties (such as well-posedness and mass conservation) of the density filter and use the obtained properties to improve the stability results.
In particular, with the new results we will be able to prove that the gradient of the density estimates is also convergent.
This is a significant property considering that the gradient is almost always required in mean-field feedback control.
In \cite{zheng2021distributed}, we presented a distributed density filter, which can be seen as a special case of the distributed density filter in this work.

The rest of the paper is organized as follows.
Section \ref{section:preliminaries} introduces some preliminaries. 
Problem formulation is given in Section \ref{section:problem formulation}. 
Section \ref{section:centralized density estimation} and \ref{section:distributed density estimation} are our main results, in which we design centralized and distributed density filters and then study their stability and optimality.
Section \ref{section:numerical implementation} discusses their numerical implementation.
Section \ref{section:simulation} performs an agent-based simulation to verify the effectiveness of the density filters.
Section \ref{section:conclusion} summarizes the contribution and points out future research.

\section{Preliminaries} \label{section:preliminaries}

\subsection{Notations}\label{section:notation}
Consider $f:E\to\mathbb{R}$, where $E\subset\mathbb{R}^n$ is a measurable set.
Denote $L^2(E)=\{f|\|f\|_{L^2(E)}<\infty\}$, endowed with the norm $\|f\|_{L^2(E)}=(\int_E|f(x)|^2dx)^{1/2}$ and inner product $\langle f,g\rangle_{L^2(E)}=\int_Ef(x)g(x)dx$. 
Let $\partial_if=\frac{\partial f}{\partial x_i}$ be the weak derivatives of $f$.
Denote $H^1(E)=\{f|\|f\|_{H^1(E)}<\infty\}$, endowed with the norm $\|f\|_{H^1(E)}=\|f\|_{L^2(E)}+\sum_{i=1}^n\|\partial_if\|_{L^2(E)}$ and inner product $\langle f,g\rangle_{H^1(E)}=\langle f,g\rangle_{L^2(E)}+\sum_{i=1}^n\langle \partial_if,\partial_ig\rangle_{L^2(E)}$.
We will omit $E$ in the notations of norms and inner products when it is clear.

Let $\Omega$ be a bounded and connected $C^1$-domain in $\mathbb R^n$ and $T>0$ be a constant. 
Denote by $\partial\Omega$ the boundary of $\Omega$.
Set $\Omega_T=\Omega\times(0,T)$. 

\subsection{Infinite-Dimensional Kalman Filters} \label{section:Kalman filter}
The Kalman filter is an algorithm that uses the system's model and sequential measurements to gradually improve its state estimate \cite{kalman1961new}. 
We present its extension to PDE systems developed in \cite{bensoussan1971filtrage}. 

Let $V$, $H$ be Hilbert spaces with norms $\|\cdot\|_V,\|\cdot\|_H$ and inner products $\langle\cdot,\cdot\rangle_V,\langle\cdot,\cdot\rangle_H$, respectively.
$V$ is dense in $H$ and $\|v\|_H\leq c\|v\|_V$ for all $v\in V$.
$V^*$ is the dual of $V$ and hence $V\subset H\subset V^*$.
Let $E$, $F$ also be Hilbert spaces.
Consider
\begin{align}\label{eq:PDE system}
&\left\{
\begin{aligned}
&\frac{du}{dt}+\mathcal{A}(t)u(t)=f(t)+\mathcal{B}(t)\xi(t), \quad t\in(0,T) \\
&u(0)=u_{0}+\zeta,
\end{aligned}
\right.\\
&\quad y(t)=\mathcal{C}(t)u(t)+\eta(t),
\end{align}
where
$$
\mathcal{A}(\cdot)\in L^\infty(0,T;\mathscr{L}(V,V^*)) \text{ and }\exists\alpha>0,\lambda\geq0 \text{ such that } 
$$
$$
\langle\mathcal{A}(t)z,z\rangle_H+\lambda\|z\|_H^{2}\geq\alpha\|z\|_V^{2},\forall z\in V, t\in (0,T),
$$
$$
\mathcal{B}(\cdot)\in L^\infty(0,T;\mathscr{L}(E,H)),\quad \mathcal{C}(\cdot)\in L^\infty(0,T;\mathscr{L}(H,F)),
$$
$$
f(\cdot)\in L^2(0,T;H),\quad u_{0},\zeta\in H
$$
$$
\xi(\cdot)\in L^2(0,T;E),\quad \eta(\cdot)\in L^2(0,T;F).
$$
These standard conditions ensure existence and uniqueness of solutions of \eqref{eq:PDE system} in the Sobolev space 
$$
W(0,T)=\{f\mid f\in L^2(0,T;V),\frac{df}{dt}\in L^2\left(0,T;V^*\right)\}.
$$

Let $\Phi:=H\times L^2(0,T;E)\times L^2(0,T;F)$ be the space of triples $\{\zeta,\xi(\cdot),\eta(\cdot)\}$.
Equip $\Phi$ with a cylindrical probability $\mu$ such that
\begin{equation*}
\operatorname{E}[\mu]=\left(\begin{array}{l}
0 \\
0 \\
0
\end{array}\right),\quad
\operatorname{Cov}[\mu]=\left(\begin{array}{ccc}
\mathcal{P}_{0} & 0 & 0 \\
0 & \mathcal{Q}(t) & 0 \\
0 & 0 & \mathcal{R}(t)
\end{array}\right)
\end{equation*}
where $\mathcal{P}_0\in\mathscr{L}(H,H)$, $\mathcal{Q}(\cdot)\in L^{\infty}(0,T;\mathscr{L}(E;E))$, and $\mathcal{R}(\cdot)\in L^{\infty}(0,T;\mathscr{L}(F;F))$ are all self-adjoint and positive semi-definite covariance operators, and $\mathcal{R}(\cdot)$ is invertible.

\begin{theorem}\label{thm:infinite-dimensional Kalman filter}
\cite{bensoussan1971filtrage} The minimum mean squared error estimator of $u$ is given by the unique solution of
\begin{align*}
&\frac{d\hat{u}}{dt}+\mathcal{A}(t)\hat{u}(t)=f(t)+\mathcal{P}(t)\mathcal{C}^*(t)\mathcal{R}^{-1}(t)\big(y(t)-\mathcal{C}(t)\hat{u}(t)\big),\\
&\hat{u}(0) = u_0,
\end{align*}
which satisfies $\hat{u}\in W(0,T)$, where $\mathcal{P}(t)\in\mathscr{L}(H;H) $ is the family of operators uniquely defined by
\begin{equation}\label{eq:regularity of P}
\left\{
\begin{aligned}
&\text{If }\varphi\in W(0,T)\text{ and }-\frac{d\varphi}{dt}+\mathcal{A}^*(t)\varphi\in L^2(0,T;H),\\
&\text{then }\mathcal{P}\varphi\in W(0,T),
\end{aligned} \right.
\end{equation}
and
\begin{equation*}
\left\{
\begin{aligned}
&\frac{d\mathcal{P}}{dt}\varphi+\mathcal{A}\mathcal{P}\varphi+\mathcal{P}\mathcal{A}^*\varphi+\mathcal{P}\mathcal{C}^*\mathcal{R}^{-1}\mathcal{C}\mathcal{P}\varphi = \mathcal{B}\mathcal{Q}\mathcal{B}^*\varphi, \\
&\mathcal{P}(0) = \mathcal{P}_0,\text{ for $\varphi$ as in \eqref{eq:regularity of P}},
\end{aligned} \right.
\end{equation*}
\end{theorem}
where $\frac{d\mathcal{P}}{dt}\varphi:=\frac{d}{dt}(\mathcal{P}\varphi)-\mathcal{P}\frac{d\varphi}{dt}$.

\subsection{Input-to-state stability}
Input-to-state stability (ISS) is a stability notion to study nonlinear control systems with external inputs \cite{sontag1989smooth}. 
Its extension to infinite-dimensional systems is developed in \cite{dashkovskiy2013input}.
Let $\left(X,\|\cdot\|_X\right)$ and $\left(U,\|\cdot\|_{U}\right)$ be the state space and the input space, endowed with norms $\|\cdot\|_X$ and $\|\cdot\|_{U}$, respectively.
Denote $U_c=PC(\mathbb{R}_+;U)$, the space of piecewise right-continuous functions from $\mathbb{R}_+$ to $U$, equipped with the sup-norm.
Define the following classes of comparison functions:
\begin{align*}
    \mathscr{P}&:=\{\gamma: \mathbb{R}_+ \to \mathbb{R}_+ \mid \gamma\text{ is continuous}, \gamma(0)=0,\\
    &\qquad\text{and }\gamma(r)>0,\forall r>0\}\\
    \mathscr{K} &:=\{\gamma\in\mathscr{P}|\gamma\text{ is strictly increasing}\}\\
    \mathscr{K}_{\infty} &:=\{\gamma\in\mathscr{K}|\gamma\text{ is unbounded}\}  \\
    \mathscr{L} &:=\{\gamma:\mathbb{R}_+\to\mathbb{R}_+|\gamma\text{ is continuous and strictly } \\
    &\qquad\text{decreasing with }\lim_{t\to\infty}\gamma(t)=0\}  \\
    \mathscr{KL} &:=\{\beta:\mathbb{R}_+\times\mathbb{R}_+\to\mathbb{R}_+|\beta(\cdot,t)\in\mathscr{K},\forall t\geq0, \\
    &\qquad\beta(r,\cdot)\in\mathscr{L},\forall r>0\}.
\end{align*}

Consider a control system $\Sigma=(X,U_c,\phi)$ where $\phi:\mathbb{R}_+\times\mathbb{R}_+\times X\times U_c\to X$ is a transition map.
Here, $\phi(t, s, x, u)$ denotes the system state at time $t \in \mathbb{R}_+$, if its state at time $s \in \mathbb{R}_+$ is $x \in X$ and the input $u \in U_c$ is applied.


\begin{definition} \label{dfn:(L)ISS}
$\Sigma$ is called \textit{locally input-to-state stable (LISS)}, if $\exists\rho_x,\rho_u>0$, $\beta\in\mathscr{KL}$, and $\gamma\in\mathscr{K}$, such that
\begin{equation*}\label{eq:(L)ISS}
    \left\|\phi\left(t,t_0,\phi_0, u\right)\right\|_X \leq \beta\left(\left\|\phi_0\right\|_X, t-t_0\right)+\gamma\Big(\sup_{t_0\leq s\leq t}\|u(s)\|_{U}\Big)
\end{equation*}
holds $\forall\phi_0:\|\phi_0\|_X\leq\rho_x,\forall u\in U_c:\|u\|_{U_c}\leq\rho_u$ and $\forall t \geq t_0$.
It is called \textit{input-to-state stable (ISS)}, if $\rho_x=\infty$ and $\rho_u=\infty$.
\end{definition} 

To emphasize the state space and input space, we will sometimes say $\|\phi(t)\|_X$ is (L)ISS to $\|u(t)\|_U$.
To verify the ISS property, Lyapunov functions can be exploited. 

\begin{definition}\label{dfn:(L)ISS-Lyapunov function}
A continuous function $V:\mathbb{R}_+\times D \to \mathbb{R}_+, D\subset X$ is called an \textit{LISS-Lyapunov function} for $\Sigma$, if $\exists\rho_x, \rho_u>0$, $\psi_{1},\psi_{2}\in\mathscr{K}_{\infty},\chi\in\mathscr{K}$, and $W\in\mathscr{P}$, such that:
\begin{itemize}
    \item[(i)] $\psi_{1}\left(\|x\|_X\right) \leq V(t, x) \leq \psi_{2}\left(\|x\|_X\right), ~\forall t \in \mathbb{R}_+, \forall x \in D$
    \item[(ii)] $\forall x \in X:\|x\|_X \leq \rho_x, \forall \xi \in U:\|\xi\|_{U} \leq \rho_u$ it holds:
    \[
    \|x\|_X \geq \chi(\|\xi\|_{U}) \Rightarrow \dot{V}_u(t,x) \leq-W(\|x\|_X), ~ \forall t\in\mathbb{R}_+
    \]
    for all $u \in U_c:\|u\|_{U_c} \leq \rho_u$ with $u(0)=\xi$, where the derivative of $V$ corresponding to the input $u$ is given by
    \[
    \dot{V}_{u}(t,x)=\varlimsup_{\delta \to+0} \frac{V(t+\delta,\phi(t+\delta,t,x,u))-V(t,x)}{\delta}.
    \]
\end{itemize}
If $D=X, \rho_x=\infty$ and $\rho_u=\infty,$ then the function $V$ is called an \textit{ISS-Lyapunov function}.
\end{definition}

An (L)ISS-Lyapunov function theorem for time-invariant systems is given in \cite{dashkovskiy2013input}.
The following theorem is an extension to time-varying systems, whose proof can be found in \cite{zheng2020transporting}.

\begin{theorem}\label{thm:(L)ISS-Lyapunov function}
\cite{zheng2020transporting} Let $\Sigma=(X,U_c,\phi)$ be a control system, and $x\equiv0$ be its equilibrium point.
If $\Sigma$ possesses an (L)ISS-Lyapunov function, then it is (L)ISS.
\end{theorem}

ISS is a convenient tool for studying the stability of cascade systems.
Consider two systems $\Sigma_i=(X_i,U_{ci},\phi_i),i=1,2$, where $U_{ci}=PC(\mathbb{R}_+;U_i)$ and $X_1\subset U_2$.
We say they form a cascade connection if $u_2(t)=\phi_1(t,t_0,\phi_{01},u_1)$.

\begin{theorem}\label{thm:(L)ISS cascade}
\cite{khalil2002nonlinear} The cascade connection of two ISS systems is ISS.
If one of them is LISS, then the cascade connection is LISS.
\end{theorem}

\section{Problem formulation} \label{section:problem formulation} 
This paper studies the problem of estimating the dynamic mean-field density of large-scale stochastic agents. 
Consider $N$ agents within a spatial domain $\Omega\subset\mathbb{R}^n$.
Their dynamics are assumed to be known and satisfy a family of stochastic differential equations:
\begin{equation} \label{eq:Langevin equation}
    dX_t^i=v(X_t^i, t) d t+\sigma(X_t^i, t) d W_t^i, \quad i = 1,\dots,N,
\end{equation}
where $X_t^i\in\Omega$ is the state of the $i$-th agent, $v=(v_{1}, \ldots, v_n)\in\mathbb{R}^n$ is the deterministic dynamics, $\{W_t^i\}_{i=1}^{N}$ is a family of $m$-dimensional standard Wiener processes independent across the agents, and $\sigma=[\sigma_{kl}]\in\mathbb{R}^{n\times m}$ is the stochastic dynamics. 
In this work, the superscription $i$ is reserved to represent the $i$-th agent.
The states $\{X_t^i\}_{i=1}^{N}$ are assumed to be observable. 
The probability density $p(x, t)$ of the states is known to satisfy a mean-field PDE, called the Fokker-Planck equation:
\begin{align}\label{eq:FP equation}
\begin{split}
    &\partial_tp=-\sum_{j=1}^{n} \partial_j(v_jp)+\sum_{j,k=1}^{n} \partial_j\partial_k(\Sigma_{jk}p) \quad\text{in }\Omega_T,\\
    &p=p_0 \quad\text{on }\Omega\times\{0\},\\
    &\mathbf{n}\cdot(g-\boldsymbol{v}p)=0 \quad\text{on }\partial\Omega\times(0,T),
\end{split}
\end{align} 
where $p_0$ is the initial density, $\Sigma=[\Sigma_{jk}]=\frac{1}{2}\sigma\sigma^{T}\in\mathbb{R}^{n\times n}$ is the diffusion tensor, $g=(g_1,\dots,g_n)$ with $g_j=\sum_{k=1}^{n}\partial_k(\Sigma_{jk}p)$, and $\mathbf{n}$ is the outward normal to $\partial\Omega$. 
By definition, $\Sigma$ is symmetric and positive semidefinite.
We always assume that it is uniformly positive definite (or uniformly elliptic), i.e., there exists a constant $\lambda>0$ such that
\begin{equation*}
    \xi^T\Sigma(x,t)\xi\geq\lambda\|\xi\|^2,\quad\forall(x,t)\in\Omega_T,\forall\xi\in\mathbb{R}^n.
\end{equation*}

\begin{remark} \label{remark:nonlinear and time-varying}
Note that \eqref{eq:FP equation} is uniquely determined by \eqref{eq:Langevin equation} which is known. 
This relationship holds even if \eqref{eq:Langevin equation} is nonlinear and time-varying.
Hence, the centralized/distributed density filters we design are applicable to a large family of systems. 
Also note that \eqref{eq:FP equation} is always a linear PDE.
\end{remark}

We recall the following result proved in \cite{zheng2020transporting} for \eqref{eq:FP equation}, which will be useful when we study the solution property of the density filters.
Note that a density function $f$ satisfies $\langle f,\mathbf{1}\rangle_{L^2}=\int_\Omega fdx=1$ where $\mathbf{1}$ is the 1-constant function on $\Omega$.

\begin{theorem}\label{thm:well-posedness}
\cite{zheng2020transporting} Assume that
\begin{equation}\label{eq:regularity condition 1}
    v_j,\sigma,\partial_j\sigma\in L^\infty(\Omega_T)\text{ and } p_0\in L^\infty(\Omega).
\end{equation}
Then
\begin{itemize}
    \item (\textbf{Well-posedness}) Problem \eqref{eq:FP equation} has a unique weak solution $p\in H^1(\Omega_T)$.
    \item (\textbf{Mass conservation}) The solution satisfies $\langle p(t),\mathbf{1}\rangle_{L^2}=1,\forall t$.
    \item (\textbf{Positivity}) If we further assume that
    \begin{equation}\label{eq:regularity condition 2}
        \partial_jv_j,\partial_j^2\sigma\in L^\infty(\Omega_T),
    \end{equation} 
    then $p_0\geq(>)0$ implies $p\geq(>)0$ for all $t\in[0,T]$.
\end{itemize}
\end{theorem}

We assume the agents can exchange information with neighbors and form a \textit{time-varying} communication topology $G(t)=(V,E(t))$, where $V$ is the set of $N$ agents and $E(t)\subset V\times V$ is the set of communication links.
We assume $G(t)$ is undirected.
Now, we can formally state the problems to be solved as follows:

\begin{problem}[Centralized density estimation] 
\label{problem:centralized density estimation}
Given the system \eqref{eq:FP equation} and agent states $\{X_t^i\}_{i=1}^{N}$, we want to estimate their density $p(x,t)$.
\end{problem}

\begin{problem}[Distributed density estimation]
\label{problem:distributed density estimation}
Given the system \eqref{eq:FP equation}, the state $X_t^i$ of an agent $i$, and the topology $G$, we want to design communication and estimation protocols for each agent to estimate the global density $p(x,t)$.
\end{problem}

\section{Centralized density estimation}

In this section, we study the centralized density estimation problem, assuming all agent states are observable.
This centralized density filter was reported in our previous work \cite{zheng2020pde}.
Our new results include studying its solution properties and using them to improve the stability results (with better regularity).
Hence, we include its development for completeness.

To design a filter, we rewrite \eqref{eq:FP equation} as an evolution equation and use KDE to construct a noisy observation $y(t)$:
\begin{equation} \label{eq:evolution equation of density}
\begin{aligned}
    \dot{p}(t) &= \mathcal{A}(t)p(t),
    \\
    y(t)&= p_{\text{KDE}}(t) = p(t)+w(t),
\end{aligned}
\end{equation}
where $\mathcal{A}(t)p=-\sum_{j=1}^{n} \partial_j(v_jp)+\sum_{j,k=1}^{n} \partial_j\partial_k(\Sigma_{jk}p)$ is a linear operator, $p_{\text{KDE}}(t)$ represents a kernel density estimator (see \eqref{eq:KDE}) using $\{X_t^i\}_{i=1}^{N}$, and $w(t)$ is the observation noise which is asymptotically Gaussian with asymptotically diagonal covariance operator $\mathcal{R}(t)=k\operatorname{diag}(p(t))$ where $\operatorname{diag}(p(t))\in\mathscr{L}(H^1(\Omega),H^1(\Omega))$ is a diagonal operator and $k>0$ is a constant depending on $N$, the kernel $K$, and its bandwidth $h$.
(See Appendices \ref{section:KDE} for these properties of $w(t)$.)
$\mathcal{R}$ is always invertible because of the positivity of $p$ in Theorem \ref{thm:well-posedness}.
The independence of $w(t)$ between different $t$ originates from the property of the Wiener process in \eqref{eq:Langevin equation}.

\label{section:centralized density estimation}
\subsection{Optimal density filter}


By Theorem \ref{thm:infinite-dimensional Kalman filter}, the optimal density filter is given by
\begin{align}
    &\dot{\hat{p}} = \mathcal{A}(t)\hat{p}+\mathcal{K}(t)(y-\hat{p}),\quad \hat{p}(0)=p_{\text{KDE}}(0),\label{eq:optimal density filter}\\
    &\mathcal{K}(t)=\mathcal{\mathcal{P}}(t)\mathcal{R}^{-1}(t),\\
    &\dot{\mathcal{P}}=\mathcal{A}(t)\mathcal{P}+\mathcal{P}\mathcal{A}^*(t)-\mathcal{P}\mathcal{R}^{-1}(t)\mathcal{P},\quad \mathcal{P}(0)=\mathcal{P}_0,\label{eq:optimal Riccati}
\end{align}
where $\mathcal{A}^*(t)$ is the adjoint operator of $\mathcal{A}(t)$, given by
\begin{equation}\label{eq:A adjoint}
    \mathcal{A}^*(t)=\sum_{j=1}^{n}v_j(x,t)\partial_j+\sum_{j,k=1}^{n}\Sigma_{jk}(x,t)\partial_j\partial_k,
\end{equation}
and \eqref{eq:optimal Riccati} should be interpreted as in Theorem \ref{thm:infinite-dimensional Kalman filter}.

First, we study the well-posedness and mass conservation properties of the density filter.
These results are new and will justify the stability proof of the density filter.

From Section \ref{section:Kalman filter}, it is seen that the well-posedness of \eqref{eq:optimal density filter} and \eqref{eq:optimal Riccati} relies on the well-posedness of \eqref{eq:FP equation}, which is already given in Theorem \ref{thm:well-posedness}.
Then, we have the following results for the solutions of \eqref{eq:optimal density filter} and \eqref{eq:optimal Riccati}.

\begin{theorem}\label{thm:well-posedness of optimal}
Assume that \eqref{eq:regularity condition 1} and \eqref{eq:regularity condition 2} hold. 
Then we have
\begin{itemize}
    \item[(i)] (\textbf{Well-posedness}) The operator Riccati equation \eqref{eq:optimal Riccati} has a unique weak solution  $\mathcal{P}(t)\in\mathscr{L}(H^1(\Omega);H^1(\Omega))$ such that $\mathcal{P}\varphi\in H^1(\Omega_T)$ for $\varphi\in H^1(\Omega_T)$ with $\frac{d\varphi}{dt}+\mathcal{A}^*(t)\varphi\in L^2(0,T;H^1(\Omega))$.
    The density filter \eqref{eq:optimal density filter} has a unique weak solution $\hat{p}\in H^1(\Omega_T)$.
    \item[(ii)] (\textbf{Mass conservation}) The solution satisfies $\langle\hat{p}(t),\mathbf{1}\rangle_{L^2}=1,\forall t$.
\end{itemize}
\end{theorem}

\begin{proof}
(i) The well-posedness follows from a direct application of Theorems \ref{thm:infinite-dimensional Kalman filter} and \ref{thm:well-posedness}. 
Note that $p\in H^1(\Omega_T)$ implies $p(t)\in H^1(\Omega)$ and $\partial_tp(t)\in H^1(\Omega)^*$ for almost all $t$.
This means we can choose $V=H=H^1(\Omega)$ and $W(0,T)=\{f\mid f\in L^2(0,T;H^1(\Omega)),\frac{df}{dt}\in L^2\left(0,T;H^1(\Omega)^*\right)\}$ for \eqref{eq:optimal density filter}.
According to Theorem \ref{thm:infinite-dimensional Kalman filter}, we obtain $\hat{p}\in H^1(\Omega_T)$ and $\mathcal{P}\in\mathscr{L}(H^1(\Omega);H^1(\Omega))$.

(ii) To prove the mass conservation property, first note that
\begin{equation*}
    \mathcal{P}_0\mathbf{1}=\operatorname{E}[(\hat{p}(0)-p(0))\circ(\hat{p}(0)-p(0))]\mathbf{1}=\mathbf{0},
\end{equation*}
where $\mathbf{0}$ is the zero constant function on $\Omega$.
Using \eqref{eq:A adjoint}, we have $\mathcal{A}^*(t)\mathbf{1}=\mathbf{0},\forall t$, and
\begin{equation*}
    \dot{\mathcal{P}}\mathbf{1}=\mathcal{A}\mathcal{P}\mathbf{1}+\mathcal{P}\mathcal{A}^{*}\mathbf{1}-\mathcal{P}\mathcal{R}^{-1}\mathcal{P}\mathbf{1}=\mathcal{A}\mathcal{P}\mathbf{1}-\mathcal{P}\mathcal{R}^{-1}\mathcal{P}\mathbf{1},
\end{equation*}
i.e., $\mathcal{P}(t)\mathbf{1}$ satisfies a homogeneous equation with zero initial condition $\mathcal{P}_0\mathbf{1}=\mathbf{0}$. Hence, $\mathcal{P}(t)\mathbf{1}=\mathbf{0}$ for all $t\geq 0$.
Now,
\begin{align*}
    \langle\dot{\hat{p}},\mathbf{1}\rangle_{L^2} &=\langle \mathcal{A}\hat{p}+\mathcal{P}\mathcal{R}^{-1}(y-\hat{p}),\mathbf{1}\rangle_{L^2} \\
    &=\langle\hat{p},\mathcal{A}^*\mathbf{1}\rangle_{L^2}+\langle \mathcal{R}^{-1}(y-\hat{p}),\mathcal{P}\mathbf{1}\rangle_{L^2} = \mathbf{0}.
\end{align*}
Hence, we have 
\begin{equation*}
    \langle\hat{p}(t),\mathbf{1}\rangle_{L^2}=\langle\hat{p}(0),\mathbf{1}\rangle_{L^2}=1\text{ for all } t\geq 0,
\end{equation*}
which means \eqref{eq:optimal density filter} conserves mass.
\end{proof}

\begin{remark}
Again, $\hat{p}\in H^1(\Omega_T)$ implies $\hat{p}(t)\in H^1(\Omega)$ for almost all $t$.
Determining that the state space is $H^1(\Omega)$ is crucial because it will enable us to prove in subsequent sections that the (distributed) density filters are convergent in $H^1$ norm, which means that not only the density estimate itself but also its gradient is convergent.
Finally, the mass conservation property of \eqref{eq:optimal Riccati} actually implies that $0$ is a simple eigenvalue of $\mathcal{P}(t)$ with $\mathbf{1}$ being an associated eigenfunction.
\end{remark}

\subsection{Suboptimal density filter}

We note that $\mathcal{R}(t)$ depends on the unknown density $p(t)$.
Hence, we approximate $\mathcal{R}(t)$ using $\bar{\mathcal{R}}(t)=\bar{k}\operatorname{diag}(p_{\text{KDE}}(t))$, where $\bar{k}$ is given by   
\begin{equation*}
    \bar{k}=\frac{h^2\int K(u)^2du+\sum_{j=1}^n\int(\partial_jK(u))^2du}{Nh^{n+2}}
\end{equation*}
(according to \eqref{eq:asymptotic normality} and \eqref{eq:derivative asymptotic normality} in Appendices).
We always assume that $\bar{\mathcal{R}}$ is invertible, which is easy to satisfy since we construct $p_{\text{KDE}}$.
Correspondingly, we obtain a ``suboptimal'' density filter:
\begin{align} 
    &\dot{\hat{p}} = \mathcal{A}(t)\hat{p}+\bar{\mathcal{K}}(t)(y(t)-\hat{p}),\quad \hat{p}(0)=p_{\text{KDE}}(0), \label{eq:suboptimal density filter}\\
    &\bar{\mathcal{K}}(t)=\bar{\mathcal{P}}(t)\bar{\mathcal{R}}^{-1}(t),\\
    &\dot{\bar{\mathcal{P}}}=\mathcal{A}(t)\bar{\mathcal{P}}+\bar{\mathcal{P}}\mathcal{A}^*(t)-\bar{\mathcal{P}}\bar{\mathcal{R}}^{-1}(t)\bar{\mathcal{P}}, \quad\bar{\mathcal{P}}(0)=\bar{\mathcal{P}}_0. \label{eq:suboptimal Riccati}
\end{align}

We have the following results for \eqref{eq:suboptimal density filter} and \eqref{eq:suboptimal Riccati}, whose proof is similar with Theorem \ref{thm:well-posedness of optimal} and thus omitted.

\begin{theorem}\label{thm:well-posedness of suboptimal}
Under assumptions \eqref{eq:regularity condition 1} and \eqref{eq:regularity condition 2}, we have
\begin{itemize}
    \item[(i)] (\textbf{Well-posedness}) The operator Riccati equation \eqref{eq:suboptimal Riccati} has a unique weak solution $\bar{\mathcal{P}}(t)\in\mathscr{L}(H^1(\Omega);H^1(\Omega))$ such that $\bar{\mathcal{P}}\varphi\in H^1(\Omega_T)$ for $\varphi\in H^1(\Omega_T)$ with $\frac{d\varphi}{dt}+\mathcal{A}^*(t)\varphi\in L^2(0,T;H^1(\Omega))$.
    The density filter \eqref{eq:suboptimal density filter} has a unique weak solution $\hat{p}\in H^1(\Omega_T)$.
    \item[(ii)] (\textbf{Mass conservation}) The solution satisfies $\langle\hat{p}(t),\mathbf{1}\rangle_{L^2}=1,\forall t$.
\end{itemize}
\end{theorem}

To study the stability of the suboptimal filter, define $\tilde{p} = \hat{p}-p$. 
Then along $\bar{{\mathcal{P}}}(t)$ we have
\begin{equation} \label{eq:estimation error equation}
    \dot{\tilde{p}} = (\mathcal{A}(t)-\bar{\mathcal{P}}(t)\bar{\mathcal{R}}^{-1}(t))\tilde{p}+\bar{\mathcal{P}}(t)\bar{\mathcal{R}}^{-1}(t)w(t).
\end{equation}
We define $\|\bar{\mathcal{R}}^{-1}(t)-\mathcal{R}^{-1}(t)\|_{H^1}$ as the approximation error of $\bar{\mathcal{R}}(t)$, since it is zero if and only if $\bar{\mathcal{R}}(t)=\mathcal{R}(t)$.
We also need the notion of generalized inverse of self-adjoint operators because of the zero eigenvalues of $\mathcal{P}$ and $\mathcal{\bar{P}}$, which is defined in Appendices \ref{section:generalized inverse}.
    
We will show that: (i) \eqref{eq:estimation error equation} is ISS to $w(t)$; (ii) the solution of \eqref{eq:suboptimal Riccati} remains close to the solution of \eqref{eq:optimal Riccati}; and (iii) the suboptimal gain $\bar{\mathcal{K}}$ remains close to the optimal gain $\mathcal{K}$, which justifies the ``suboptimality'' of \eqref{eq:suboptimal density filter}.
In \cite{zheng2020pde}, we proved weaker results assuming that all density filters and operator Riccati equations have unique solutions in certain Hilbert spaces.
In this work, we use the newly obtained solution properties and the notion of generalized inverse operators to prove stability results in a better functional space (with improved regularity).
The stability and optimality results are stated in the following.

\begin{theorem} \label{thm:stability of centralized density filter}
Assume that \eqref{eq:regularity condition 1} and \eqref{eq:regularity condition 2} hold, and that $\|\mathcal{P}(t)\|_{H^1}$ and $\|\bar{\mathcal{P}}(t)\|_{H^1}$ are uniformly bounded. 
Also assume that there exist constants $c_1,c_2>0$ such that for all $t\geq0$,
\begin{equation}\label{eq:uniform positivity centralized}
    0<c_1\mathcal{I}\leq \mathcal{R}^{-1}(t),\bar{\mathcal{R}}^{-1}(t),\mathcal{P}^{\dagger}(t),\bar{\mathcal{P}}^{\dagger}(t)\leq c_2\mathcal{I},
\end{equation}
where $\mathcal{P}^{\dagger}(t)$ ($\bar{\mathcal{P}}^{\dagger}(t)$) is the generalized inverse of $\mathcal{P}(t)$ ($\mathcal{\bar{P}}(t)$) and $\mathcal{I}$ is the identity operator on $H^1(\Omega)$.
Then we have
\begin{itemize}
    \item[(i)] (\textbf{Stability}) $\|\hat{p}(t)-p(t)\|_{H^1}$ is ISS to $\|w(t)\|_{H^1}$.
    \item[(ii)] $\|\bar{\mathcal{P}}(t)-\mathcal{P}(t)\|_{H^1}$ is LISS to $\|\bar{\mathcal{R}}^{-1}(t)-\mathcal{R}^{-1}(t)\|_{H^1}$.
    \item[(iii)] (\textbf{Optimality}) $\|\bar{\mathcal{K}}(t)-\mathcal{K}(t)\|_{H^1}$ is LISS to $\|\bar{\mathcal{R}}^{-1}(t)-\mathcal{R}^{-1}(t)\|_{H^1}$.
\end{itemize}
\end{theorem}

\begin{proof}
(i) We study the unforced part of \eqref{eq:estimation error equation}, i.e., assume $w=0$.
Consider a Lyapunov function $V_1=\langle\bar{\mathcal{P}}^{\dagger}(t)\tilde{p}(t),\tilde{p}(t)\rangle_{H^1}$.
By the definition of $\bar{\mathcal{P}}^\dagger(t)$, we see that $0$ is a simple eigenvalue of $\bar{\mathcal{P}}^\dagger(t)$ for all $t$, with $\mathbf{1}$ being an associated eigenfunction.
Then $V_1=0$ if and only if $\tilde{p}=\text{constant}$, which would necessarily be 0 because $\langle\tilde{p},\mathbf{1}\rangle_{L^2}=\langle\hat{p}-p,\mathbf{1}\rangle_{L^2}=0$ (the mass conservation property).
By appropriate approximation arguments, we have
\begin{align*}
    \dot{V}_1 &= \big\langle\bar{\mathcal{P}}^\dagger\dot{\tilde{p}},\tilde{p}\big\rangle_{H^1} + \big\langle\bar{\mathcal{P}}^\dagger\tilde{p},\dot{\tilde{p}}\big\rangle_{H^1} - \big\langle\bar{\mathcal{P}}^\dagger\dot{\bar{\mathcal{P}}}\bar{\mathcal{P}}^\dagger\tilde{p},\tilde{p}\big\rangle_{H^1}\\
    &= \big\langle\bar{\mathcal{P}}^\dagger(A-\bar{\mathcal{P}}\bar{R}^{-1})\tilde{p},\tilde{p}\big\rangle_{H^1} + \big\langle(\mathcal{A}^*-\bar{\mathcal{R}}^{-1}\bar{\mathcal{P}})\bar{\mathcal{P}}^\dagger\tilde{p},\tilde{p}\big\rangle_{H^1}\\
    & \quad -\big\langle \bar{\mathcal{P}}^\dagger(\mathcal{A}\bar{\mathcal{P}}+\bar{\mathcal{P}}\mathcal{A}^{*}-\bar{\mathcal{P}}\bar{\mathcal{R}}^{-1}\bar{\mathcal{P}})\bar{\mathcal{P}}^\dagger\tilde{p},\tilde{p} \big\rangle_{H^1}\\
    &= -\langle \bar{\mathcal{R}}^{-1}\tilde{p},\tilde{p} \rangle_{H^1}.
\end{align*}
Hence, the unforced part of \eqref{eq:estimation error equation} is uniformly exponentially stable. 
Since \eqref{eq:estimation error equation} is a linear control system, with the uniform boundedness of $\bar{\mathcal{P}}$ and $\bar{\mathcal{R}}^{-1}$, we obtain that $\|\hat{p}(t)-p(t)\|_{H^1}$ is ISS to $\|w(t)\|_{H^1}$.





(ii) Define $\Gamma=\bar{{\mathcal{P}}}-{\mathcal{P}}$. 
Using \eqref{eq:optimal Riccati} and \eqref{eq:suboptimal Riccati} we have
\begin{equation} \label{eq:Riccati error equation1}
\begin{aligned}
    &\dot{\Gamma} = \mathcal{A}(t)\Gamma+\Gamma \mathcal{A}^*(t) - \bar{{\mathcal{P}}}\bar{\mathcal{R}}^{-1}(t)\bar{{\mathcal{P}}} + {\mathcal{P}} \mathcal{R}^{-1}(t){\mathcal{P}}.
\end{aligned}
\end{equation}
Similarly, by considering a Lyapunov function defined by $V_2=\langle\mathcal{P}^\dagger(t)\tilde{p}(t),\tilde{p}(t)\rangle_{H^1}$, one can show that the following system along $\mathcal{P}(t)$ is also uniformly exponentially stable:
\begin{equation} \label{eq:unforced estimation error equation2}
    \dot{\tilde{p}} = (\mathcal{A}-\mathcal{P}\mathcal{R}^{-1})\tilde{p}.
\end{equation}
Note that the inner products in $V_1$ and $V_2$ are equivalent because of assumption \eqref{eq:uniform positivity centralized}.
We rewrite \eqref{eq:Riccati error equation1} as
\begin{align*}
    \dot{\Gamma}
    &=\mathcal{A}\Gamma+\Gamma \mathcal{A}^{*} - \bar{\mathcal{P}}\bar{\mathcal{R}}^{-1}\Gamma - \bar{\mathcal{P}} \bar{\mathcal{R}}^{-1}\mathcal{P} - \Gamma \mathcal{R}^{-1}\mathcal{P} + \bar{\mathcal{P}}\mathcal{R}^{-1}\mathcal{P}\\
    &=(\mathcal{A}-\bar{\mathcal{P}}\bar{\mathcal{R}}^{-1})\Gamma + \Gamma(\mathcal{A}^{*}-\mathcal{R}^{-1}\mathcal{P}) - \bar{\mathcal{P}}(\bar{\mathcal{R}}^{-1}-\mathcal{R}^{-1})\mathcal{P}
\end{align*}
which is essentially a linear equation. 
Now fix $q\in H^1(\Omega)$ with $\|q\|_{H^1}=1$. 
Since the unforced part of \eqref{eq:estimation error equation} and \eqref{eq:unforced estimation error equation2} are uniformly exponentially stable, and $\|\bar{\mathcal{P}}(t)\|_{H^1}$ and $\|\mathcal{P}(t)\|_{H^1}$ are uniformly bounded, $\exists$ constants $\lambda,c>0$ such that
\begin{align*}
    \|\Gamma(t)q\|_{H^1}
    &\leq e^{-\lambda t}\|\Gamma(0)q\|_{H^1}\\
    &\quad+\frac{c}{\lambda}\sup_{0\leq\tau\leq t}\|\bar{R}^{-1}(\tau)-R^{-1}(\tau)\|_{H^1},
\end{align*}
i.e., $\|\bar{\mathcal{P}}(t)-\mathcal{P}(t)\|_{H^1}$ is LISS to $\|\Bar{\mathcal{R}}^{-1}(t)-\mathcal{R}^{-1}(t)\|_{H^1}$.

(ii) To prove the third statement, observe that
\begin{align*}
    &\quad\|\Bar{\mathcal{K}}(t)-\mathcal{K}(t)\|_{H^1}\\ &=\|(\Bar{\mathcal{P}}\Bar{\mathcal{R}}^{-1}-\mathcal{P}\mathcal{R}^{-1})(t)\|_{H^1} \\
    &\leq \|(\Bar{\mathcal{P}}\Bar{\mathcal{R}}^{-1}-\mathcal{P}\Bar{\mathcal{R}}^{-1} + \mathcal{P}\Bar{\mathcal{R}}^{-1}-\mathcal{P}\mathcal{R}^{-1})(t)\|_{H^1} \\
    &\leq \|\Bar{\mathcal{R}}^{-1}(t)\|_{H^1}\|\Bar{\mathcal{P}}(t)-\mathcal{P}(t)\|_{H^1}\\
    &\quad+\|\mathcal{P}(t)\|_{H^1}\|\Bar{\mathcal{R}}^{-1}(t)-\mathcal{R}^{-1}(t)\|_{H^1}.
\end{align*}
By Theorem \ref{thm:(L)ISS cascade}, $\|\Bar{\mathcal{K}}(t)-\mathcal{K}(t)\|_{H^1}$ is LISS to $\|\Bar{\mathcal{R}}^{-1}(t)-\mathcal{R}^{-1}(t)\|_{H^1}$.
\end{proof}

\begin{remark}
A few comments are in order.
By taking advantage of the dynamics, the density filter essentially combines all past outputs (constructed using KDE) to produce better and convergent estimates, which largely circumvent the problem of bandwidth selection of KDE \cite{silverman1986density}.
The fact that $\hat{p}(t)$ is convergent in $H^1(\Omega)$ implies that the $L^2$ norm of its gradient is also convergent.
This implication is significant as it is observed that the existing mean-field feedback laws are more or less based on \textit{gradient flows} \cite{villani2008optimal}, which use the gradient of density estimates as feedback \cite{eren2017velocity, krishnan2018distributed, elamvazhuthi2018pde, zheng2020transporting}.
Since the gradient operator is an unbounded operator, any algorithm that produces accurate density estimates may have arbitrarily large estimation error for the gradient.
However, by formulating a filtering problem and studying its solution property, we are able to obtain density estimates whose gradient is also convergent.
\end{remark}





\section{Distributed density estimation}
\label{section:distributed density estimation}
In this section, we study the distributed density estimation problem.
We present a distributed density filter by integrating communication protocols into \eqref{eq:suboptimal density filter} and \eqref{eq:suboptimal Riccati}, and study its convergence and optimality.
We reformulate \eqref{eq:evolution equation of density} in the following distributed form:
\begin{equation} \label{eq:distributed evolution equation of density}
\begin{aligned}
    &\Dot{p}(t) = \mathcal{A}(t)p(t),
    \\
    &z_i(t) =K_i(t), \quad i=1,\dots,N,
\end{aligned}
\end{equation}
where $K_i(t)=\frac{1}{h^n}K\big(\frac{1}{h}(x-X_t^i)\big)$ is a kernel density function centered at position $X_t^i$.
We view $z_i(t)$ as the local observation made by the $i$-th agent.
The challenge of distributed density estimation lies in that each agent alone does not have any meaningful observation of the unknown density $p(t)$, because its local observation $z_i(t)$ is simply a kernel centered at position $X_t^i$, which conveys no information about $p(t)$.

\subsection{Local density filter}

Inspired by the distributed Kalman filters in \cite{olfati2007distributed}, we design the local density filter for the $i$-th agent ($i=1\dots,N$) as
\begin{align} 
    &\Dot{\Hat{p}}_i = \mathcal{A}(t)\Hat{p}_i+\mathcal{L}_i(t)(y_i(t)-\Hat{p}_i)+\theta\mathcal{P}_i(t)\sum_{j\in\mathcal{N}_i}(\hat{p}_j-\hat{p}_i), \label{eq:local density filter}\\
    &\mathcal{K}_i(t) = \mathcal{P}_i(t)\mathcal{R}_i(t), \\
    &\dot{\mathcal{P}}_i=\mathcal{A}(t)\mathcal{P}_i+\mathcal{P}_i\mathcal{A}^*(t)-\mathcal{P}_i\mathcal{R}_i^{-1}(t)\mathcal{P}_i, \label{eq:local Riccati}\\
    &\Hat{p}_i(0)=y_i(0),\quad \mathcal{P}_i(0)=\mathcal{P}_{i0},
\end{align}
where $\theta\geq0$, $y_i(t)$ is to be constructed later, $\mathcal{R}_i(t)=\Bar{k}\operatorname{diag}(y_i(t))$, and $\mathcal{N}_i$ represents the neighbors of agent $i$.
The algorithm we reported in \cite{zheng2021distributed} can be seen as a special case ($\theta=0$) of the above algorithm.

An important observation is that the quantities in \eqref{eq:suboptimal density filter} satisfy
\begin{align*}
    y(t)=\frac{1}{N}\sum_{i=1}^Nz_i(t),\quad\Bar{\mathcal{R}}(t)=\operatorname{diag}(\frac{1}{N}\sum_{i=1}^Nz_i(t)).
\end{align*}
These relations suggest that if we can design algorithms such that $y_i(t)\to\frac{1}{N}\sum_i^Nz_i(t)$ (which then implies $\mathcal{R}_i(t)\to\Bar{\mathcal{R}}(t)$), then the local filter \eqref{eq:local density filter} and \eqref{eq:local Riccati} should converge to the centralized filter \eqref{eq:suboptimal density filter} and \eqref{eq:suboptimal Riccati}.
The associated problem of tracking the average of $N$ time-varying signals is called dynamic average consensus.
We use a variant of the proportional-integral (PI) consensus estimator given in \cite{freeman2006stability}, which is a low-pass filter and therefore suitable for our problem.
Some useful properties are summarized in Appendices \ref{section:dynamic average consensus}.
According to \eqref{eq:PI consensus estimator}, we construct $y_i$ in the following way:
\begin{align}\label{eq:PI consensus estimator for kernels}
\begin{split}
    &\partial_t\psi_i =-\alpha(\psi_{i}-z_{i})\\
    &\qquad\quad-\sum_{j\in\mathcal{N}_i}a_{ij}(\psi_{i}-\psi_{j})+\sum_{j\in\mathcal{N}_i} b_{ji}(\phi_{i}-\phi_{j}),\\
    &\partial_t\phi_i =-\sum_{j\in\mathcal{N}_i} b_{ij}(\psi_{i}-\psi_{j}),\\
    &\psi_i(0)=z_i(0), \quad\phi_i(0)=\phi_{i0},\\
    &y_i = \frac{\psi_i+d_i}{\|\psi_i+d_i\|_{L^1}},
\end{split}
\end{align}
where $\alpha,a_{ij},b_{ij}>0$ are constants, $\phi_{i0}$ needs to be a density function, and $d_i$ is defined by
\begin{equation*}
    d_i =
    \begin{cases}
    |c-\inf_x\psi_i(x)| & \text{if }\inf_x\psi_i(x)<c,\\
    0 & \text{if }\psi_i(x)\geq c,
    \end{cases}
\end{equation*}
where $c>0$ is a pre-specified small constant.
The first two equations of \eqref{eq:PI consensus estimator for kernels} are based on \eqref{eq:PI consensus estimator}, i.e., the consensus algorithm is performed in a pointwise manner.
The last equation is added to redistribute the negative mass of $\psi_i$ to other spatial areas such that the output becomes positive and conserves mass. 
Note that $\langle y_i,\mathbf{1}\rangle_{L^2}=1$, and $y_i=\psi_i$ if $\psi_i\geq c,\forall x$.
Hence, as long as there exists $t_c>0$ such that $\psi_i(t)\geq c$ for $t\geq t_c$ (which should be expected because $\psi_i$ is meant to track the strictly positive density function $p_{\text{KDE}}(t)$), we can assume $y_i=\psi_i$ in the convergence analysis.

First, we show that the internal states of \eqref{eq:PI consensus estimator for kernels} conserve mass.

\begin{theorem}
The solution of \eqref{eq:PI consensus estimator for kernels} satisfies $\langle \psi_i(t),\mathbf{1}\rangle_{L^2}=\langle \phi_i(t),\mathbf{1}\rangle_{L^2}=1,\forall i,\forall t\geq 0$.

\end{theorem}

\begin{proof}
First, since $z_i(t)=K_i(t)$, we have $\langle z_{i}(t),\mathbf{1}\rangle_{L^2}=1,\forall t$.
Now consider change of variables: $\bar{\psi}_i=\psi_i-\mathbf{1}$ and $\bar{\phi}_i=\phi_i-\mathbf{1}$.
Then $\langle\bar{\psi}_i(0),\mathbf{1}\rangle_{L^2}=\langle\bar{\phi}_i(0),\mathbf{1}\rangle_{L^2}=0$, and
\begin{align*}
\begin{split}
    \langle\partial_t\bar{\psi}_i,\mathbf{1}\rangle_{L^2} &= -\alpha\langle \bar{\psi}_i,\mathbf{1}\rangle_{L^2}-\sum_{j=1}^{N} a_{i j}\langle \bar{\psi}_i-\bar{\psi}_j,\mathbf{1}\rangle_{L^2}\\
    &\quad+\sum_{j=1}^{N} b_{j i}\langle \bar{\phi}_{i}-\bar{\phi}_{j},\mathbf{1}\rangle_{L^2}, \\
    \langle\partial_t\bar{\phi}_i,\mathbf{1}\rangle_{L^2} &=-\sum_{j=1}^{N} b_{i j}\langle \bar{\psi}_i-\bar{\psi}_j,\mathbf{1}\rangle_{L^2},
\end{split}
\end{align*}
i.e., $(\langle\bar{\psi}_i(t),\mathbf{1}\rangle_{L^2},\langle\bar{\phi}_i(t),\mathbf{1}\rangle_{L^2})$ satisfies a family of homogeneous equations with zero initial conditions.
Hence, $\langle\bar{\psi}_i(t),\mathbf{1}\rangle_{L^2}=\langle\bar{\phi}_i(t),\mathbf{1}\rangle_{L^2}=0$ and $\langle\psi_i(t),\mathbf{1}\rangle_{L^2}=\langle\phi_i(t),\mathbf{1}\rangle_{L^2}=1$ for all $t\geq0$.
\end{proof}

Now we show that the local quantities indeed track the central quantities.
Define $\mathbf{z}(x,t)=(z_1(x,t),\dots,z_N(x,t))$ and $\mathbf{\Pi}=\mathbf{I}_N-\frac{1}{N}\mathbf{1}_N\mathbf{1}_N^\top$ where $\mathbf{I}_N\in\mathbb{R}^{N\times N}$ is the identity matrix and $\mathbf{1}_N=(1,\dots,1)\in\mathbb{R}^N$.
We have the following theorem.

\begin{theorem}\label{thm:ISS of PI consensus estimator for kernels}
Assume the communication topology $G(t)$ is always connected.
If $\exists t_c>0$ such that $\psi_i(t)\geq c$ for $t\geq t_c$, then $\|y_i(t)-y(t)\|_{H^1}$ and $\|\mathcal{R}_i(t)-\Bar{\mathcal{R}}(t)\|_{H^1}$ are ISS to $\|\mathbf{\Pi}\dot{\mathbf{z}}(t)\|_{H^1}$.
Moreover, if $\|\mathcal{R}_i^{-1}(t)\|_{H^1}$ and $\|\Bar{\mathcal{R}}^{-1}(t)\|_{H^1}$ are uniformly bounded, then $\|\mathcal{R}_i^{-1}(t)-\Bar{\mathcal{R}}^{-1}(t)\|_{H^1}$ is also ISS to $\|\mathbf{\Pi}\dot{\mathbf{z}}(t)\|_{H^1}$.
\end{theorem}


\begin{proof}
According to Lemma \ref{lmm:ISS of PI estimator} in Appendices \ref{section:dynamic average consensus} (with proper generalizations to infinite-dimensional ODEs), $\|y_i(t)-y(t)\|_{L^2}$ is ISS to $\|\mathbf{\Pi}\dot{\mathbf{z}}(t)\|_{L^2}$.
By taking the partial derivatives in $x$ on both sides of \eqref{eq:PI consensus estimator for kernels}, we see that their spatial derivatives satisfy the same equations and hence $\|\partial_jy_i(t)-\partial_jy(t)\|_{L^2}$ is ISS to $\|\mathbf{\Pi}\partial_j\dot{\mathbf{z}}(t)\|_{L^2},\forall j=1,\dots,n$.
As a result, $\|y_i(t)-y(t)\|_{H^1}$ is ISS to $\|\mathbf{\Pi}\dot{\mathbf{z}}(t)\|_{H^1}$, which also implies that $\|\mathcal{R}_i(t)-\Bar{\mathcal{R}}(t)\|_{H^1}$ is ISS to $\|\mathbf{\Pi}\dot{\mathbf{z}}(t)\|_{H^1}$.
Finally, notice that
\begin{align*}
    &\|\mathcal{R}_i^{-1}(t)-\Bar{\mathcal{R}}^{-1}(t)\|_{H^1}\\
    \leq&\|\mathcal{R}_i^{-1}(t)\|_{H^1}\|\Bar{\mathcal{R}}^{-1}(t)\|_{H^1}\|\mathcal{R}_i(t)-\Bar{\mathcal{R}}(t)\|_{H^1}
\end{align*}
which implies $\|\mathcal{R}_i^{-1}(t)-\Bar{\mathcal{R}}^{-1}(t)\|_{H^1}$ is also ISS.
\end{proof}

\begin{remark}
The above ISS property holds even if the network is switching, as long as it remains connected.
In practice, the network may not be always connected since the agents are mobile.
Nevertheless, the transient error caused by permanent dropout of agents will be slowly forgotten according to \eqref{eq:robust to initialization}.
The mass conservation property is also useful in numerical implementation because it helps to avoid error accumulation caused by numerical errors or agent dropout.
\end{remark}

\subsection{Stability and optimality}

Similar to the centralized filters, we have the following results for the local density filter \eqref{eq:local density filter} and the local operator Riccati equation \eqref{eq:local Riccati}.

\begin{theorem}\label{thm:well-posedness of local}
Assume that \eqref{eq:regularity condition 1} and \eqref{eq:regularity condition 2} hold, and that $\exists t_c>0$ such that $\psi_i(t)\geq c$ for $t\geq t_c$. 
We have
\begin{itemize}
    \item[(i)] (\textbf{Well-posedness}) The local operator Riccati equation \eqref{eq:local Riccati} has a unique weak solution $\mathcal{P}_i(t)\in\mathscr{L}(H^1(\Omega);H^1(\Omega))$ such that $\mathcal{P}_i\varphi\in H^1(\Omega_T)$ for $\varphi\in H^1(\Omega_T)$ with $\frac{d\varphi}{dt}+\mathcal{A}^*(t)\varphi\in L^2(0,T;H^1(\Omega))$.
    The local density filter \eqref{eq:local density filter} has a unique weak solution $\hat{p}_i\in H^1(\Omega_T)$. 
    \item[(ii)] (\textbf{Mass conservation}) The solution satisfies $\langle\hat{p}_i(t),\mathbf{1}\rangle_{L^2}=1,\forall t$.
\end{itemize}
\end{theorem}


\begin{proof}
(i) The well-posedness of \eqref{eq:local Riccati} for each $i$ follows from Theorems \ref{thm:infinite-dimensional Kalman filter} and \ref{thm:well-posedness}.
The local density filter \eqref{eq:local density filter} is a family of $N$ PDEs where the coupling only occurs on the low-order terms.
Hence, it lies in the framework of \eqref{eq:PDE system} and Theorem \ref{thm:infinite-dimensional Kalman filter}, and has a unique weak solution $\hat{p}_i\in H^1(\Omega_T)$.

(ii) By similar arguments as in Theorem \ref{thm:well-posedness of optimal}, we have $\mathcal{P}_i(t)\mathbf{1}=\mathbf{0},\forall t\geq0$.
Then,
\begin{align*}
    &\quad\langle\dot{\hat{p}}_i,\mathbf{1}\rangle_{L^2}\\ &=\langle \mathcal{A}\hat{p}_i+\mathcal{P}_i\mathcal{R}_i^{-1}(y_i-\hat{p}_i)+\theta\mathcal{P}_i\sum_{j\in\mathcal{N}_i}(\hat{p}_j-\hat{p}_i),\mathbf{1}\rangle_{L^2} \\
    &=\langle\hat{p}_i,\mathcal{A}^*\mathbf{1}\rangle_{L^2}+\langle \mathcal{R}_i^{-1}(y_i-\hat{p}_i),\mathcal{P}_i\mathbf{1}\rangle_{L^2}\\
    &\quad+\theta\langle\sum_{j\in\mathcal{N}_i}(\hat{p}_j-\hat{p}_i),\mathcal{P}_i\mathbf{1}\rangle_{L^2} = \mathbf{0}.
\end{align*}
Hence, $\langle\hat{p}_i(t),\mathbf{1}\rangle_{L^2}=\langle\hat{p}_i(0),\mathbf{1}\rangle_{L^2}=1$ for all $t\geq 0$.
\end{proof}

The remaining task is to show that \eqref{eq:local density filter} and \eqref{eq:local Riccati} indeed remain close to \eqref{eq:suboptimal density filter} and \eqref{eq:suboptimal Riccati}, respectively.
Towards this end, define $\tilde{p}_i = \hat{p}_i-p$.
Note that $\hat{p}_{j}-\hat{p}_{i}=\tilde{p}_j-\tilde{p}_{i}$.
Then along $\mathcal{P}_i(t)$ we have
\begin{equation*}
    \dot{\tilde{p}}_i=(\mathcal{A}-\mathcal{P}_i\mathcal{R}_i^{-1})\tilde{p}_i+\mathcal{P}_i\mathcal{R}_i^{-1}(y_i-y+w) +\theta\mathcal{P}_i\sum_{j\in\mathcal{N}_i}(\tilde{p}_{j}-\tilde{p}_i).
\end{equation*}

\begin{theorem} \label{thm:stability of local density filter}
Let the communication topology $G(t)$ be connected.
Assume that \eqref{eq:regularity condition 1} and \eqref{eq:regularity condition 2} hold, and that $\exists t_c>0$ such that $\psi_i(t)\geq c$ for $t\geq t_c$.
Also assume that $\|\mathcal{P}_i(t)\|_{H^1}$ and $\|\bar{\mathcal{P}}(t)\|_{H^1}$ are uniformly bounded and that there exist constants $c_1,c_2>0$ such that for all $t\geq 0$,
\begin{equation*}
    0<c_1\mathcal{I}\leq \mathcal{R}_i^{-1}(t),\bar{\mathcal{R}}^{-1}(t),\mathcal{P}_i^{\dagger}(t),\bar{\mathcal{P}}^{\dagger}(t)\leq c_2\mathcal{I},
\end{equation*}
where $\mathcal{P}_i^{\dagger}(t)$ ($\bar{\mathcal{P}}^{\dagger}(t)$) is the generalized inverse of $\mathcal{P}_i(t)$ ($\mathcal{\bar{P}}(t)$).
Then we have:
\begin{itemize}
    \item[(i)] (\textbf{Stability}) $\sum_{i=1}^N\|\hat{p}_i(t)-p(t)\|_{H^1}$ is ISS to $\sum_{i=1}^N\|y_i(t)-y(t)\|_{H^1}$ and $\|w(t)\|_{H^1}$.
    \item[(ii)] $\|\mathcal{P}_i(t)-\bar{\mathcal{P}}(t)\|_{H^1}$ is LISS to $\|\mathcal{R}_i^{-1}(t)-\bar{\mathcal{R}}^{-1}(t)\|_{H^1}$.
    \item[(iii)] (\textbf{Optimality}) $\|\mathcal{K}_i(t)-\bar{\mathcal{K}}(t)\|_{H^1}$ is LISS to $\|\mathcal{R}_i^{-1}(t)-\bar{\mathcal{R}}^{-1}(t)\|_{H^1}$. 
\end{itemize}
\end{theorem}

\begin{proof}
(i) Consider a Lyapunov function $V=\sum_{i=1}^{N}\langle\mathcal{P}_i^{\dagger}(t)\tilde{p}_i(t),\tilde{p}_i(t)\rangle_{H^1}$.
Then $V=0$ if and only if $\tilde{p}_i=0,\forall i$.
We have
\begin{align*}
    \dot{V}&=\sum_{i=1}^N\Big(\langle\mathcal{P}_i^{\dagger}\dot{\tilde{p}}_i,\tilde{p}_i\rangle_{H^1} + \langle\mathcal{P}_i^{\dagger}\tilde{p}_i,\dot{\tilde{p}}_i\rangle_{H^1} - \langle\mathcal{P}_i^{\dagger}\dot{\mathcal{P}}_i\mathcal{P}_i^{\dagger}\tilde{p}_i,\tilde{p}_i\rangle_{H^1}\Big)\\
    &=\sum_{i=1}^N\Big(\langle\mathcal{P}_i^{\dagger}(\mathcal{A}-\mathcal{P}_i\mathcal{R}_i^{-1})\tilde{p}_i,\tilde{p}_i\rangle_{H^1}\\
    &\quad+ \langle\mathcal{R}_i^{-1}(y_i-y+w),\tilde{p}_i\rangle_{H^1}\\
    &\quad+\langle\mathcal{P}_i^{\dagger}\tilde{p}_i,(\mathcal{A}-\mathcal{P}_i\mathcal{R}_i^{-1})\tilde{p}_i\rangle_{H^1}\\
    &\quad+ \langle\mathcal{P}_i^{\dagger}\tilde{p}_i,\mathcal{P}_i\mathcal{R}_i^{-1}(y_i-y+w)\rangle_{H^1}\\
    &\quad-\langle\mathcal{P}_i^{\dagger}(\mathcal{A}\mathcal{P}_i+\mathcal{P}_i\mathcal{A}^*-\mathcal{P}_i\mathcal{R}_i^{-1}\mathcal{P}_i)\mathcal{P}_i^{\dagger}\tilde{p}_i,\tilde{p}_i\rangle_{H^1}\\
    &\quad+\langle\theta\sum_{j\in\mathcal{N}_i}(\tilde{p}_{j}-\tilde{p}_i),\tilde{p}_i\rangle_{H^1}\\
    &\quad+ \langle\mathcal{P}_i^{\dagger}\tilde{p}_i,\theta\mathcal{P}_i\sum_{j\in\mathcal{N}_i}(\tilde{p}_{j}-\tilde{p}_i)\rangle_{H^1}\Big)\\
    &=\sum_{i=1}^N\Big(-\langle\mathcal{R}_i^{-1}\tilde{p}_i,\tilde{p}_i\rangle_{H^1} + 2\langle\mathcal{R}_i^{-1}(y_i-y+w),\tilde{p}_i\rangle_{H^1}\Big)\\
    &\quad+\theta\sum_{i=1}^N\sum_{j\in\mathcal{N}_i}\langle(\tilde{p}_{j}-\tilde{p}_i),\tilde{p}_i\rangle_{H^1},
\end{align*}
where the last term is negative according to consensus theory \cite{olfati2007distributed}.
Let $\gamma\in(0,1)$.
Then,
\begin{align*}
    \dot{V}&\leq\sum_{i=1}^N\Big(-(1-\gamma)\langle\mathcal{R}_i^{-1}\tilde{p}_i,\tilde{p}_i\rangle_{H^1} - \gamma\langle\mathcal{R}_i^{-1}\tilde{p}_i,\tilde{p}_i\rangle_{H^1}\\
    &\quad+2\langle\mathcal{R}_i^{-1}(y_i-y+w),\tilde{p}_i\rangle_{H^1}\Big)\\
    &\leq-(1-\gamma)\sum_{i=1}^N\langle\mathcal{R}_i^{-1}\tilde{p}_i,\tilde{p}_i\rangle_{H^1},
\end{align*}
if
\begin{equation*}
    \sum_{i=1}^N\|\tilde{p}_i\|_{H^1}\geq\frac{2c_2}{\gamma c_1}\sum_{i=1}^N\|y_i-y+w\|_{H^1}.
\end{equation*}
Since $\|y_i-y+w\|_{H^1}\leq\|y_i-y\|_{H^1}+\|w\|_{H^1}$, we obtain that $\sum_{i=1}^N\|\hat{p}_i(t)-p(t)\|_{H^1}$ is ISS to $\sum_{i=1}^N\|y_i(t)-y(t)\|_{H^1}$ and $\|w(t)\|_{H^1}$.

(ii) Define $\Gamma_i=\bar{\mathcal{P}}-\mathcal{P}_i$. 
We have
\begin{equation*}
    \dot{\Gamma}_i = \mathcal{A}\Gamma_i+\Gamma_i\mathcal{A}^{*} - \mathcal{P}\bar{\mathcal{R}}^{-1}\bar{\mathcal{P}} + \mathcal{P}_i\mathcal{R}_i^{-1}\mathcal{P}_i,
\end{equation*}
which is formally equivalent to \eqref{eq:Riccati error equation1}.
Hence we omit the proof.

(iii) The proof is also similar to Theorem \ref{thm:stability of centralized density filter} (iii).
\end{proof}

Finally, by Theorem \ref{thm:(L)ISS cascade}, we have the following corollary.

\begin{corollary}
Under the assumptions in Theorems \ref{thm:ISS of PI consensus estimator for kernels} and \ref{thm:stability of local density filter}, we have:
\begin{itemize}
    \item[(i)] (\textbf{Stability}) $\sum_{i=1}^N\|\hat{p}_i(t)-p(t)\|_{H^1}$ is ISS to $\|\mathbf{\Pi}\dot{\mathbf{z}}(t)\|_{H^1}$ and $\|w(t)\|_{H^1}$.
    \item[(ii)] $\|\mathcal{P}_i(t)-\bar{\mathcal{P}}(t)\|_{H^1}$ is LISS to $\|\mathbf{\Pi}\dot{\mathbf{z}}(t)\|_{H^1}$.
    \item[(iii)] (\textbf{Optimality}) $\|\mathcal{K}_i(t)-\bar{\mathcal{K}}(t)\|_{H^1}$ is LISS to $\|\mathbf{\Pi}\dot{\mathbf{z}}(t)\|_{H^1}$. 
\end{itemize}
\end{corollary}

\begin{remark}
The distributed density filter allows each agent to estimate the global density of all agents using only its own position and local information exchange.
Like the centralized density filter, the density estimates and their gradient are both convergent.
Notice that we allow $\theta=0$ in \eqref{eq:local density filter}, in which case each agent only communicates $z_i=K_i$ with neighbors.
Thus, it only needs to exchange its real-time position $X_t^i$ with neighbors since $K_i$ is uniquely determined by $X_t^i$, which is very efficient from a communication perspective but may result in slow convergence.
When $\theta>0$, each agent additionally exchange its density estimate $\hat{p}_i$ which results in faster convergence by using extra synchronization but also increases communication cost.
The accelerated convergence can be observed from the proof of Theorem \ref{thm:stability of local density filter} where the $\theta$ term provides additional negative mass to $\dot{V}$.
\end{remark}

\section{Numerical implementation}
\label{section:numerical implementation}
In this section, we discuss the numerical implementation of the proposed algorithms.
The simplest implementation is probably the one based on finite difference approximations of derivatives \cite{strikwerda2004finite}, although more delicate numerical methods are also possible.
Assume the spatial domain $\Omega$ is discretized as a grid with $M$ cells.
Then the density functions like $\hat{p}_i$ are approximated by $M\times1$ vectors.
The operators like $\mathcal{A}$, $\mathcal{P}_i$ and $\mathcal{R}_i$ are approximated by $M\times M$ matrices.
The time evolution is based on Euler discretization.
Hence, \eqref{eq:local density filter} and \eqref{eq:local Riccati} can be approximated by:
\begin{align*}
    \hat{p}_i(k+1)&=(I+\delta_TA(k))\hat{p}(k)\\
    &\quad+\delta_TP_i(k)R_i(k)(y_i(k)-p_i(k))\\
    &\quad+\delta_T\theta P_i(k)\sum_{j\in\mathcal{N}_i}(\hat{p}_j(k)-\hat{p}_i(k)),\\
    P_i(k+1)&=(I+\delta_TA)P_i(k)+\delta_TP_i(k)A^\top(k)\\
    &\quad-\delta_TP_i(k)R_i^{-1}(k)P_i(k),
\end{align*}
where $\hat{p}_i,y_i\in\mathbb{R}^M$, $A,P_i,R_i\in\mathbb{R}^{M\times M}$, and $\delta_T$ is the step size.
The consensus algorithm \eqref{eq:PI consensus estimator for kernels} can be approximated in a similar way, with a step size $\delta_t$.
The selection of $\delta_t$ has been studied in many papers \cite{kia2019tutorial}.
We may let $\delta_T=l\delta_t$, i.e., perform multiple steps of information exchange between two successive updates of the filters.
With this type of asynchronous updates, the local filters can potentially achieve better tracking of the centralized filter.
This setup of asynchronous updates is left as future work.
In terms of computational cost, the computational bottleneck of Kalman filters is on the inverse of the covariance operators/matrices.
We also notice that $M$ is usually large in the discretization scheme.
However, $R_i$ is diagonal and $A$ is highly sparse.
Hence, the numerical iteration of the filter is very efficient and can be carried out online.
In terms of communication cost, when $\theta=0$, each agent only sends its own position to neighbors during each iteration, which is very efficient.
If $\theta>0$, each agent additionally sends an $M\times M$ matrix.
In this case, each agent may decrease the resolution/dimension of $\hat{p}_i$ by downsampling before sending it to neighbors, and increase the resolution of a received density estimate by interpolation.

\section{Simulation studies}
\label{section:simulation}

\begin{figure*}[t]
    \centering
    \begin{subfigure}[b]{0.24\textwidth}
        \centering
        \includegraphics[width=1\textwidth]{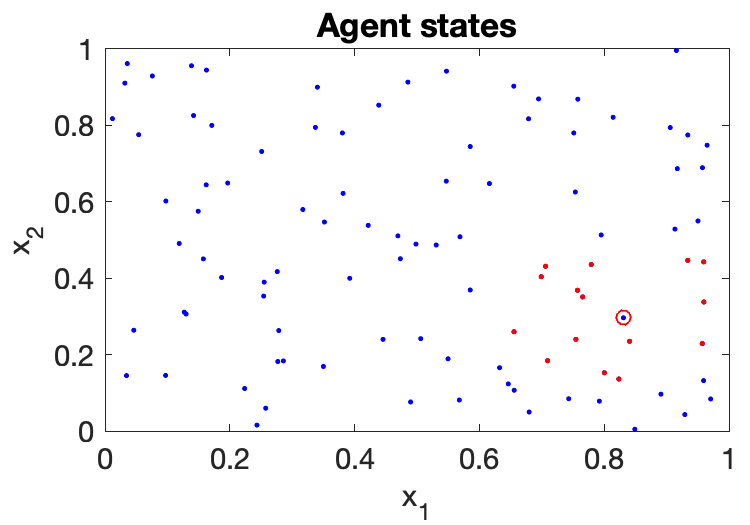}
    \end{subfigure}
    \begin{subfigure}[b]{0.24\textwidth}
        \centering
        \includegraphics[width=1\textwidth]{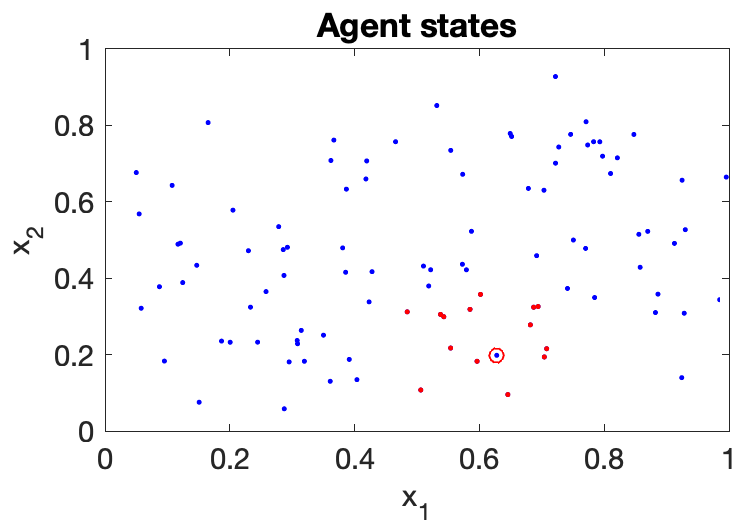}
    \end{subfigure}
    \begin{subfigure}[b]{0.24\textwidth}
        \centering
        \includegraphics[width=1\textwidth]{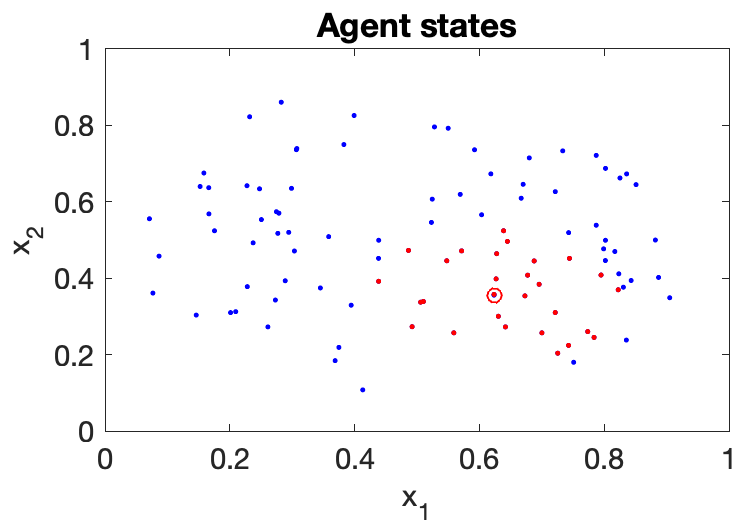}
    \end{subfigure}
    \begin{subfigure}[b]{0.24\textwidth}
        \centering
        \includegraphics[width=1\textwidth]{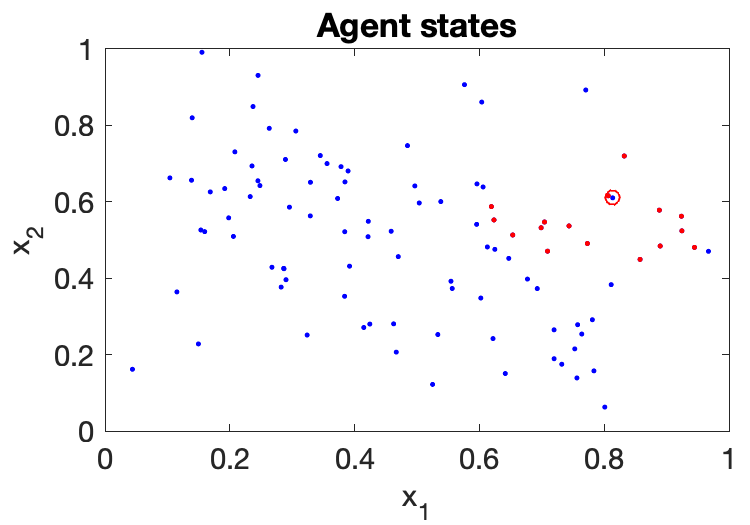}
    \end{subfigure}
    \vspace{3mm}
    
    \begin{subfigure}[b]{0.24\textwidth}
        \centering
        \includegraphics[width=1\textwidth]{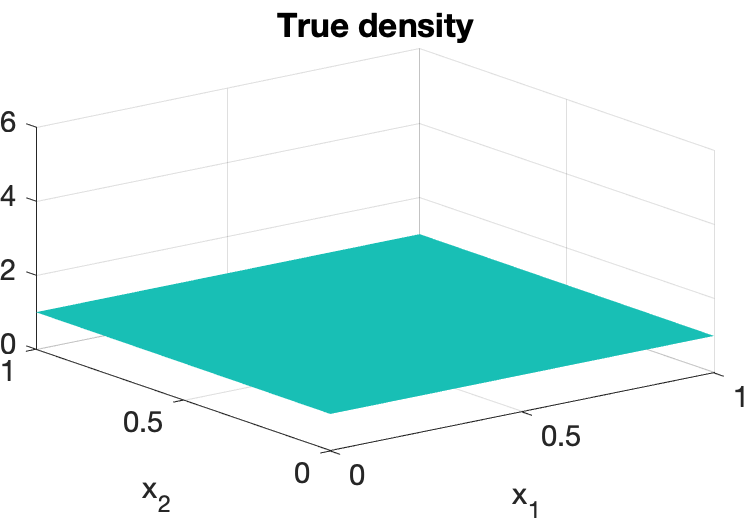}
    \end{subfigure}
    \begin{subfigure}[b]{0.24\textwidth}
        \centering
        \includegraphics[width=1\textwidth]{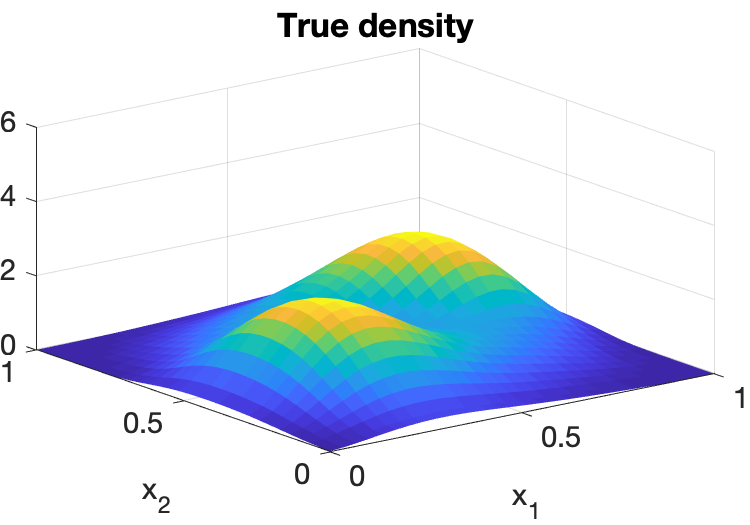}
    \end{subfigure}
    \begin{subfigure}[b]{0.24\textwidth}
        \centering
        \includegraphics[width=1\textwidth]{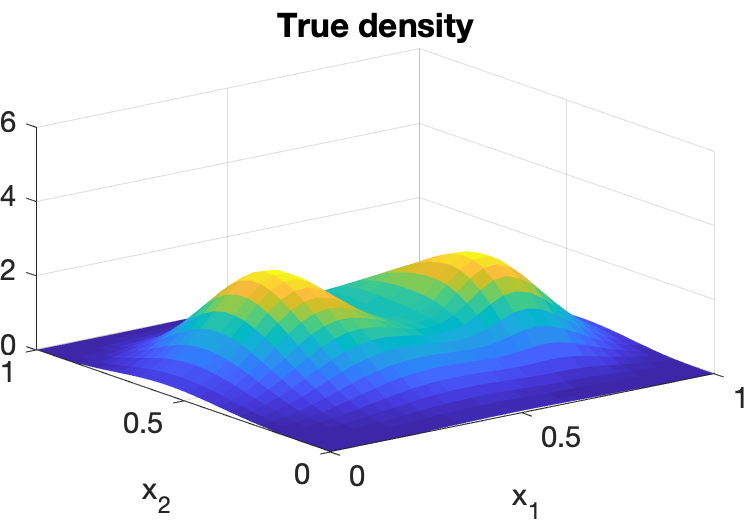}
    \end{subfigure}
    \begin{subfigure}[b]{0.24\textwidth}
        \centering
        \includegraphics[width=1\textwidth]{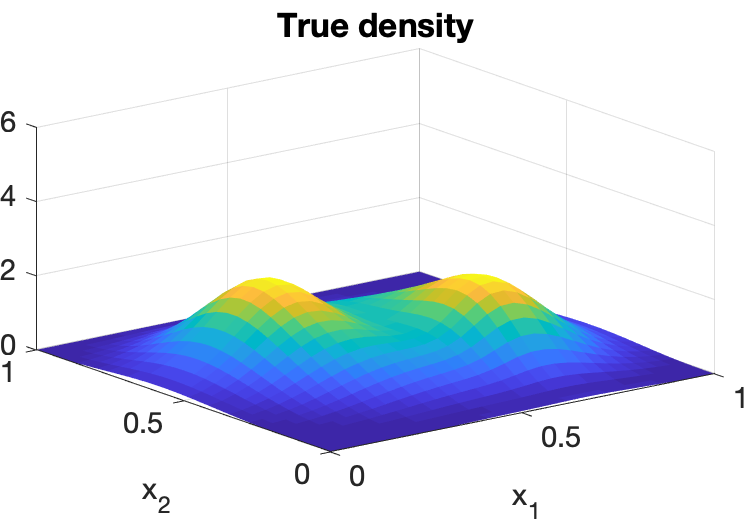}
    \end{subfigure}
    \vspace{3mm}

    \begin{subfigure}[b]{0.24\textwidth}
        \centering
        \includegraphics[width=1\textwidth]{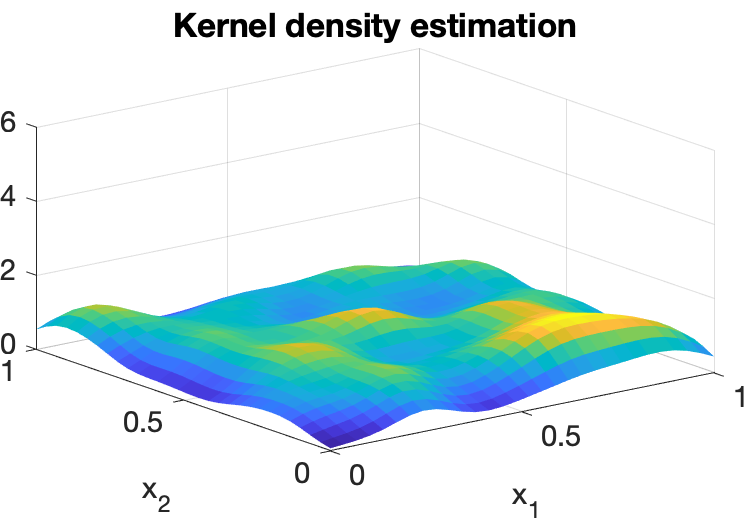}
    \end{subfigure}
    \begin{subfigure}[b]{0.24\textwidth}
        \centering
        \includegraphics[width=1\textwidth]{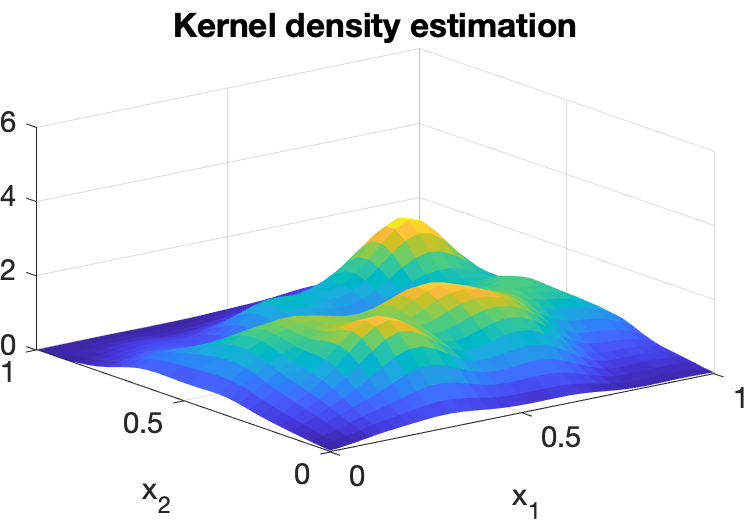}
    \end{subfigure}
    \begin{subfigure}[b]{0.24\textwidth}
        \centering
        \includegraphics[width=1\textwidth]{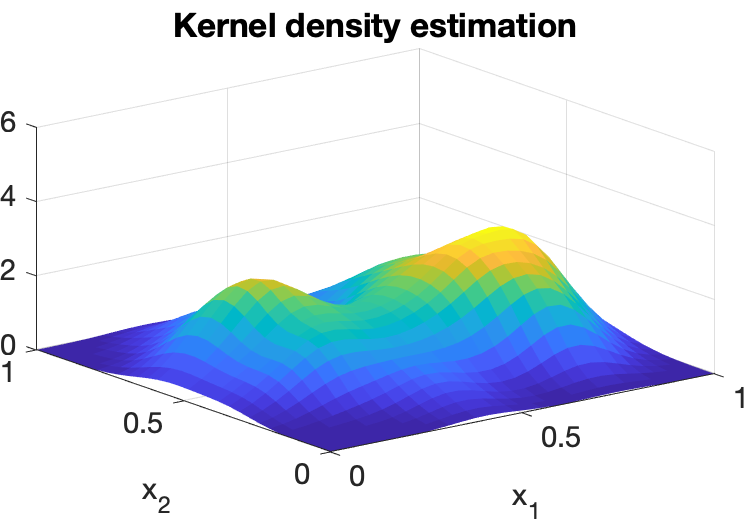}
    \end{subfigure}
    \begin{subfigure}[b]{0.24\textwidth}
        \centering
        \includegraphics[width=1\textwidth]{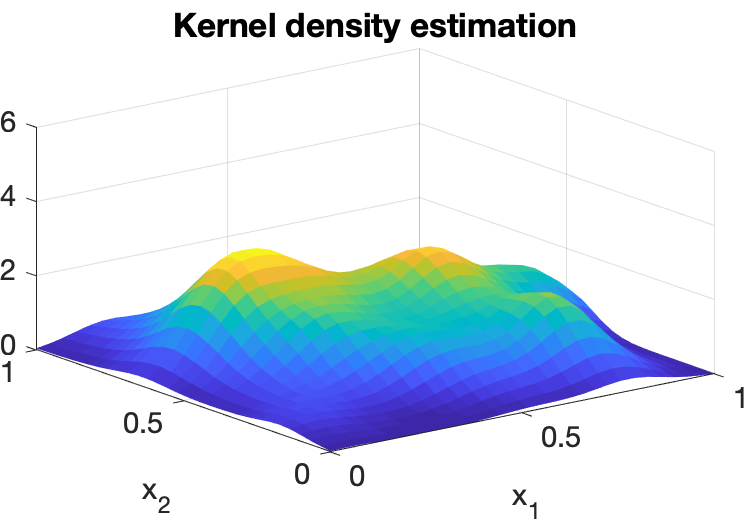}
    \end{subfigure}
    \vspace{3mm}

    \begin{subfigure}[b]{0.24\textwidth}
        \centering
        \includegraphics[width=1\textwidth]{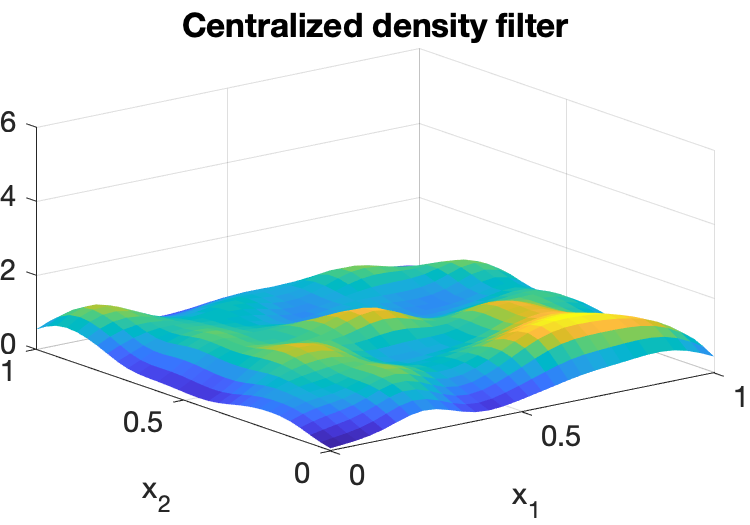}
    \end{subfigure}
    \begin{subfigure}[b]{0.24\textwidth}
        \centering
        \includegraphics[width=1\textwidth]{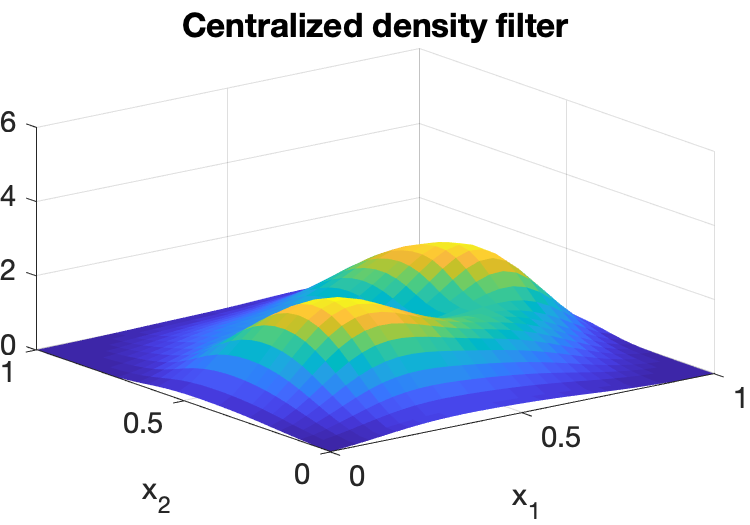}
    \end{subfigure}
    \begin{subfigure}[b]{0.24\textwidth}
        \centering
        \includegraphics[width=1\textwidth]{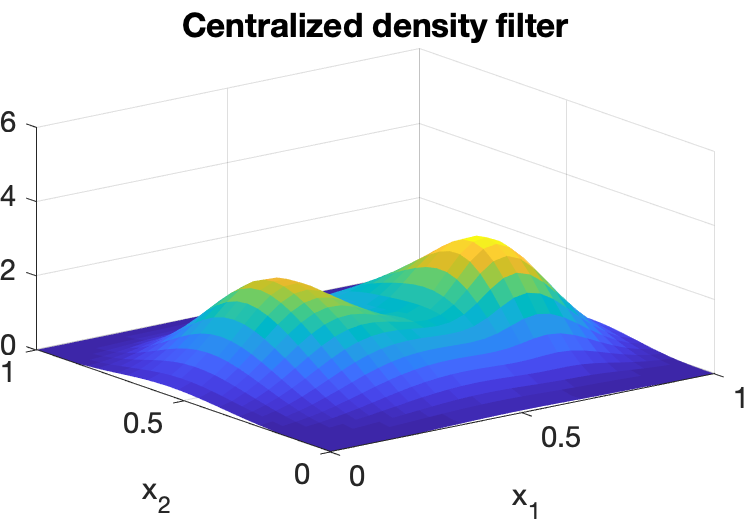}
    \end{subfigure}
    \begin{subfigure}[b]{0.24\textwidth}
        \centering
        \includegraphics[width=1\textwidth]{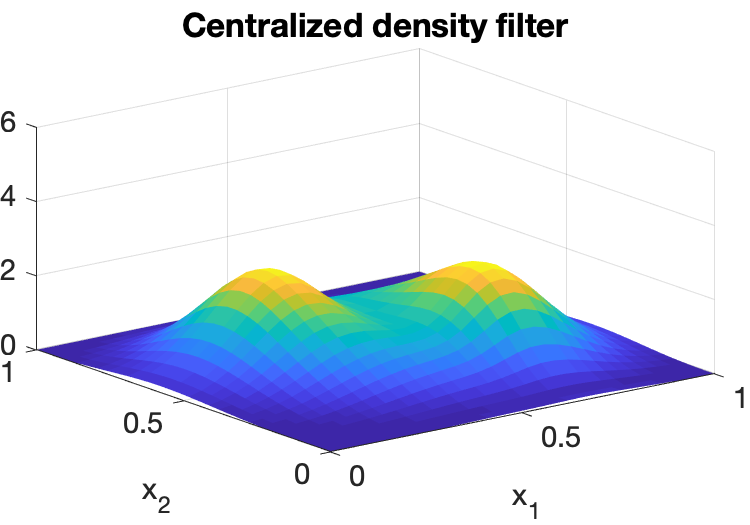}
    \end{subfigure}
    \vspace{3mm}
    
    \begin{subfigure}[b]{0.24\textwidth}
        \centering
        \includegraphics[width=1\textwidth]{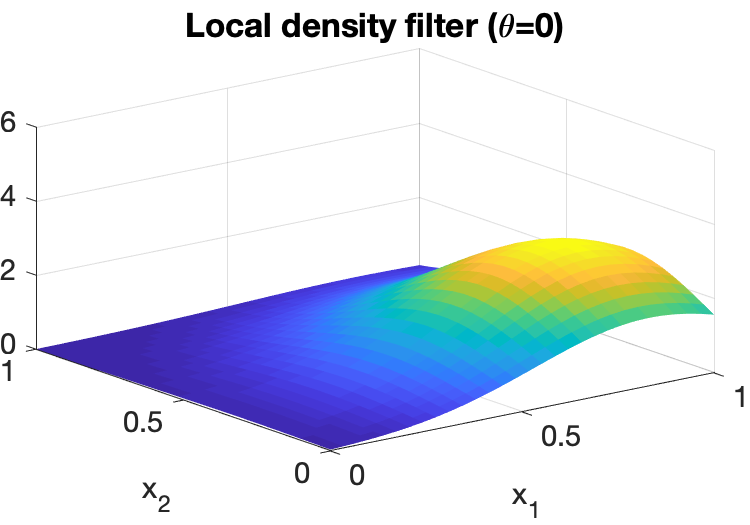}
    \end{subfigure}
    \begin{subfigure}[b]{0.24\textwidth}
        \centering
        \includegraphics[width=1\textwidth]{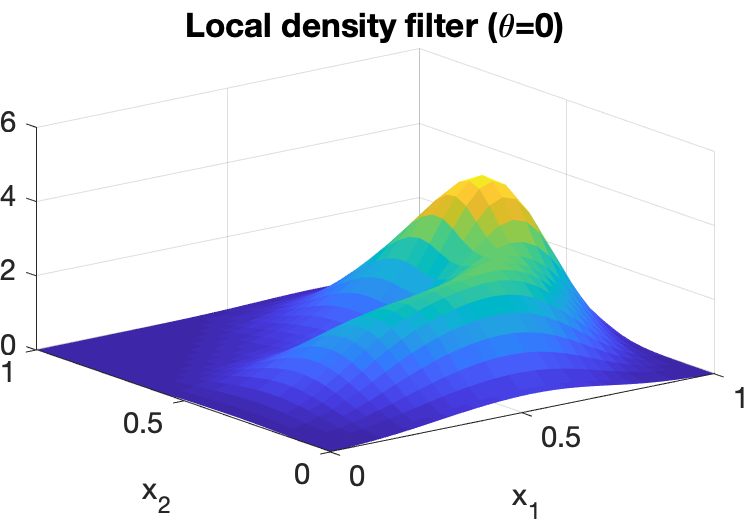}
    \end{subfigure}
    \begin{subfigure}[b]{0.24\textwidth}
        \centering
        \includegraphics[width=1\textwidth]{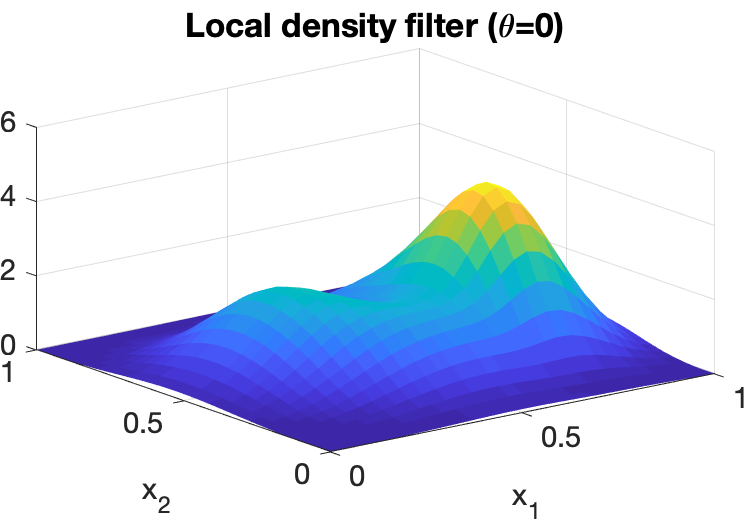}
    \end{subfigure}
    \begin{subfigure}[b]{0.24\textwidth}
        \centering
        \includegraphics[width=1\textwidth]{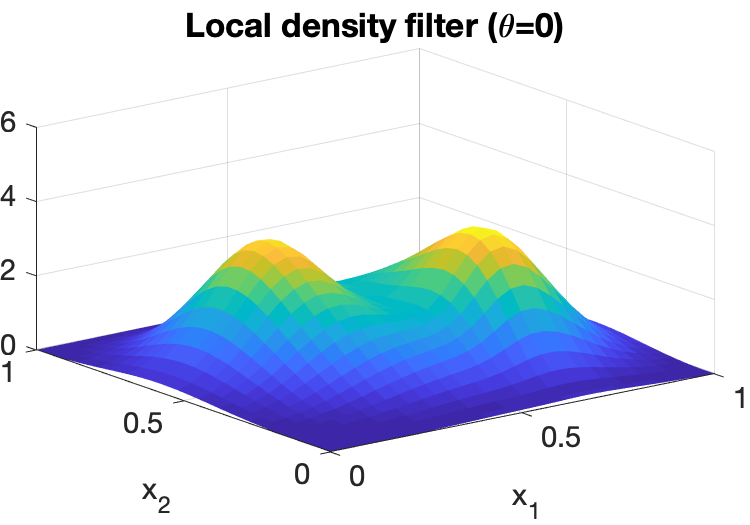}
    \end{subfigure}
    \vspace{3mm}
    
    \begin{subfigure}[b]{0.24\textwidth}
        \centering
        \includegraphics[width=1\textwidth]{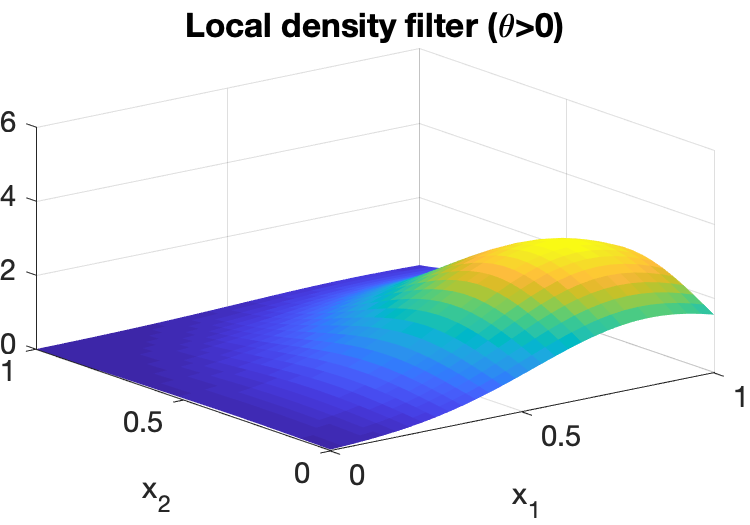}
    \end{subfigure}
    \begin{subfigure}[b]{0.24\textwidth}
        \centering
        \includegraphics[width=1\textwidth]{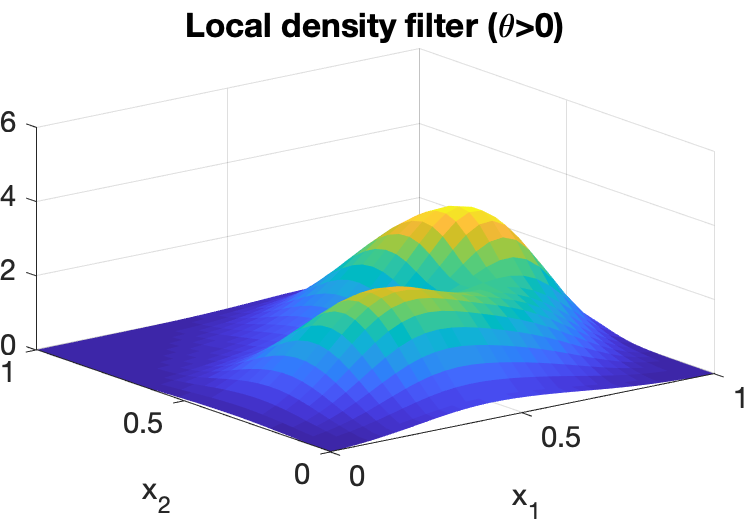}
    \end{subfigure}
    \begin{subfigure}[b]{0.24\textwidth}
        \centering
        \includegraphics[width=1\textwidth]{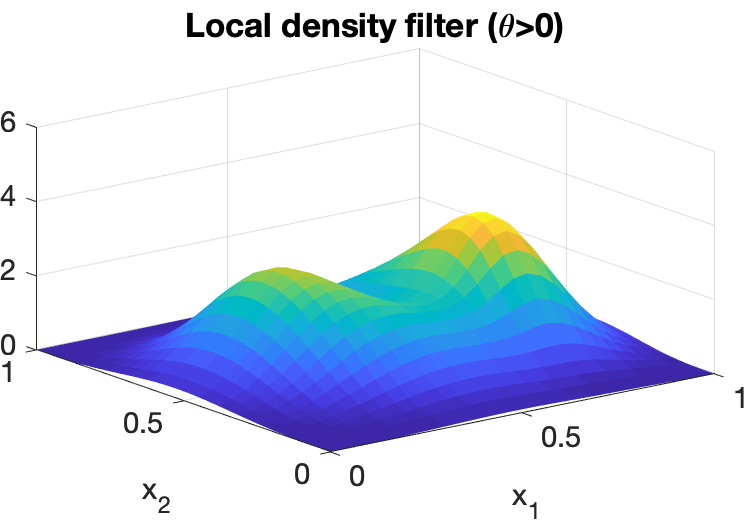}
    \end{subfigure}
    \begin{subfigure}[b]{0.24\textwidth}
        \centering
        \includegraphics[width=1\textwidth]{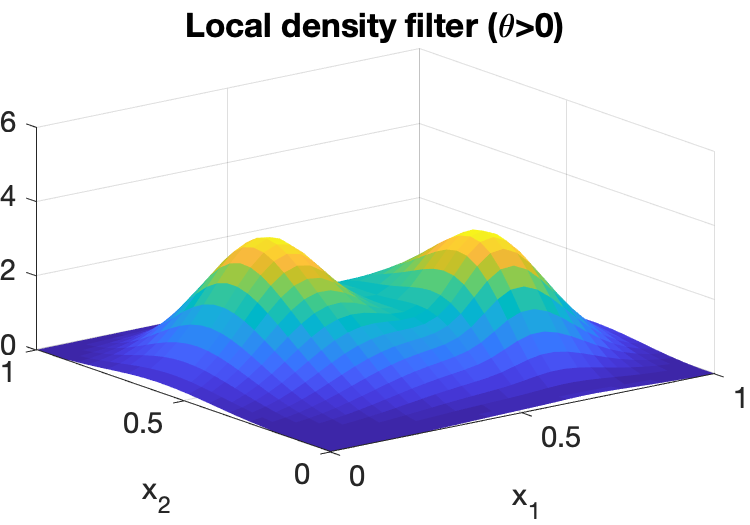}
    \end{subfigure}
    \caption{\textbf{1st row:} the collection of states $\{X_t^i\}_{i=1}^{100}$ of the systems \eqref{eq:example of Langevin equation}, where the red circle is the randomly selected representative agent and the red dots are its neighbors;
    \textbf{2nd row:} the evolution of the ground truth density $p(x,t)$ computed using \eqref{eq:example of FP equation};
    \textbf{3rd row:} density estimates $p_{\text{KDE}}(x,t)$ using KDE;
    \textbf{4th row:} outputs $\hat{p}(x,t)$ of the centralized filter \eqref{eq:suboptimal density filter};
    \textbf{5th row:} outputs $\hat{p}_i(x,t)$ of the local filter \eqref{eq:local density filter} when $\theta=0$;
    \textbf{6th row:} outputs $\hat{p}_i(x,t)$ of the local filter \eqref{eq:local density filter} when $\theta>0$.}
    \label{fig:density filter}
\end{figure*}

\begin{figure}[hbt!]
    \centering
    \begin{subfigure}[b]{1\columnwidth}
        \centering
        \includegraphics[width=0.9\columnwidth]{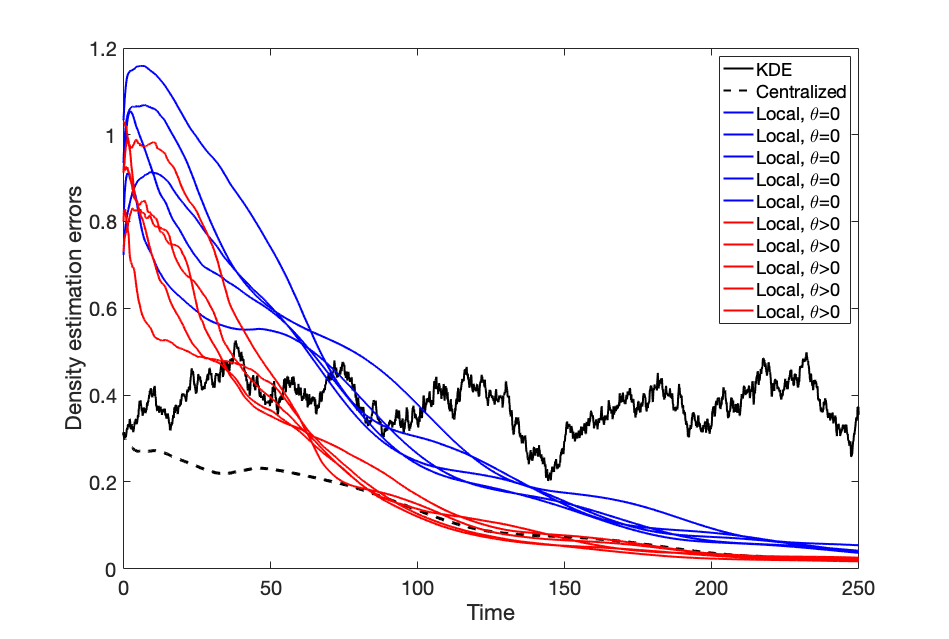}
        \caption{Density estimation errors.}
    \end{subfigure}
    
    \begin{subfigure}[b]{1\columnwidth}
        \centering
        \includegraphics[width=0.9\columnwidth]{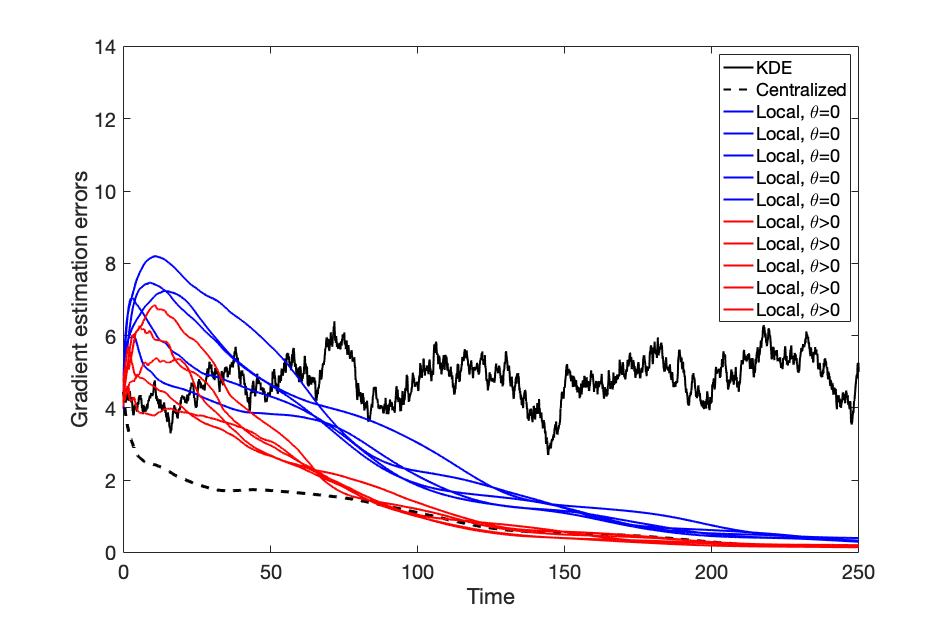}
        \caption{Gradient estimation errors.}
    \end{subfigure}
    \caption{\textbf{Black solid line:} density/gradient estimation errors of KDE; \textbf{black dashed line:} density/gradient estimation errors of the centralized density filter; \textbf{blue line:} density/gradient estimation errors of local filters when $\theta=0$; \textbf{red line:} density/gradient estimation errors of local filters when $\theta>0$.}
    \label{fig:density error}
\end{figure}

To verify the performance of the proposed centralized/distributed filters, we perform a simulation on Matlab using $100$ agents.
The agents' positions are generated by the following family of equations:
\begin{equation}\label{eq:example of Langevin equation}
    dX_t^i=\nabla\cdot\frac{D\nabla f(x,t)}{f(x,t)} dt+DdW_t^i, \quad i = 1,\dots,100,
\end{equation}
where $D=0.03$ and $f(x,t)$ is a time-varying density function to be specified. 
We set $\Omega=(0,1)^2$ with a reflecting boundary.
The initial positions $\{X_0^i\}_{i=1}^{100}$ are drawn from the uniform distribution. 
Therefore, the ground truth density satisfies
\begin{align} \label{eq:example of FP equation}
\begin{split}
     &\partial_t p =-\nabla\cdot\frac{Dp\nabla f}{f} + \frac{1}{2}D^2\Delta p, \\
     &p(\cdot,0)=\mathbf{1},
\end{split}
\end{align}
which will be used for comparison.
We design $f(x,t)$ to be a Gaussian mixture of two components with the common covariance matrix $\text{diag}(0.02,0.02)$ but different time-varying means $[0.5+0.3\cos(0.04t),0.5+0.3\sin(0.04t)]^T$ and $[0.5+0.3\cos(0.04t+\pi),0.5+0.3\sin(0.04t+\pi)]^T$. 
Under this design, the agents are nonlinear and time-varying and their states will converge to two ``spinning'' Gaussian components.
 
The numerical implementation is based on the procedures outlined in Section \ref{section:numerical implementation}.
Specifically, $\Omega$ is discretized as a $30\times30$ grid.
All density functions are approximated by $900\times1$ vectors and all operators are approximated by $900\times900$ matrices.
The density updates and information exchange are performed in a synchronous manner with a step size $\delta_T=0.1s$.
The bandwidth of KDE is $h=0.08$. 
For the PI consensus algorithm \eqref{eq:PI consensus estimator for kernels}, we set $\alpha=0.2$, $a_{ij}=0.4$, and $b_{ij}=0.04$.

A graphical illustration of the simulation results is presented in Fig. \ref{fig:density filter}, where each column represents the same time step and each row represents the evolution of an investigated quantity. 
As already observed in \cite{zheng2020pde}, the centralized filter, by taking advantage of the dynamics, quickly "recognizes" the two underlying Gaussian components and gradually catches up with the evolution of the ground truth density, outperforming the outputs of KDE.
Then, we randomly select an agent and investigate the performance of its local density filter when $\theta=0$ and $\theta=0.4$.
We observe that the two local filters also gradually catch up with the ground truth density, however in a slower rate compared with the centralized filter due to the delay effect of information exchange.
In Fig. \ref{fig:density error}, we compare the $L^2$ norms of density and gradient estimation errors of KDE, the centralized filter, and 5 randomly selected local filters.
We observe that all the filters outperform KDE as the filtering process proceeds.
Both the density estimates and the gradient estimates are convergent.
The local filters have larger initial density/gradient estimation errors because of limited observation on the initial step.
Through information exchange, they gradually produce comparable performance with the centralized filter.
In particular, the one with $\theta>0$ converges faster than the one with $\theta=0$ because of the additional exchange of density estimates with neighbors.

\section{Conclusion}\label{section:conclusion}
We presented centralized/distributed density filters for estimating the mean-field density of large-scale agent systems with known dynamics by a novel integration of mean-field models, KDE, infinite-dimensional Kalman filters, and communication protocols. 
With the distributed filter, each agent was able to estimate the global density using only the mean-field dynamics, its own state/position, and local information exchange with its neighbors.
It was scalable to the number of agents, convergent in not only the density estimates but also their gradient, and efficient in terms of numerical implementation and communication.
These algorithms can be used for many density-based distributed optimization and control problems of large-scale agent systems when density feedback information is required.
Our future work is to integrate the density filters into mean-field feedback control for large-scale agent systems and study the closed-loop stability.

\section*{Appendices}

\subsection{Kernel density estimation}
\label{section:KDE}
Kernel density estimation is a non-parametric way to estimate an unknown probability density function.
Let $X_1,\dots,X_N\in\mathbb{R}^n$ be independent identically distributed random variables having a common probability density $f$.
The kernel density estimator for $f$ is given by \cite{parzen1962estimation, silverman1986density}
\begin{equation}\label{eq:KDE}
    f_N(x) = \frac{1}{Nh^{n}}\sum_{i=1}^{N}K\Big(\frac{1}{h}(x-X_{i})\Big),
\end{equation}
where $K(x)$ is a kernel function \cite{silverman1986density}
and $h$ is the bandwidth, usually chosen as a function of $N$ such that $\lim_{N\to\infty}h=0$ and $\lim_{N\to\infty}Nh^n=\infty$.
The Gaussian kernel is frequently used, given by
$$
K(x)=\frac{1}{(2\pi)^{n/2}}\exp\Big(-\frac{1}{2}x^\top x\Big).
$$
It is known that the $f_N(x)$ is asymptotically normal and that $f_N(x_i)$ and $f_N(x_j)$ are asymptotically uncorrelated for $x_i\neq x_j$, as $N\to\infty$, which is summarized as follows.

\begin{lemma}
[Asymptotic normality \cite{cacoullos1966estimation}]
\label{lmm:asymptotic normality}
If $\lim_{N\to\infty}h=0$, then at any continuity point $x$,
\begin{equation} \label{eq:asymptotic unbisedness}
    \lim_{N\to\infty}\operatorname{E}[f_N(x)]-f(x)=0.
\end{equation}
If in addition $\lim_{N\to\infty}Nh^n=\infty$, then as $N\to\infty$,
\begin{equation} \label{eq:asymptotic normality}
    \sqrt{Nh^{n}}(f_N(x)-\operatorname{E}[f_N(x)])\to\mathcal{N}\Big(0,f(x)\int K(u)^{2} d u\Big)
\end{equation}
in distribution, where $\mathcal{N}$ represents the Gaussian distribution.
\end{lemma}

\begin{lemma}[Asymptotic uncorrelatedness \cite{cacoullos1966estimation}]
\label{lmm:asymptotic uncorrelatedness}
Let $x_i,x_j$ be two distinct continuity points.
If $\lim_{N\to\infty}h=0$, then the covariance of $f_N(x_i)$ and $f_N(x_j)$ satisfies
$$
Nh^{n}\operatorname{Cov}[f_N(x_i), f_N(x_j)] \to 0,\text{ as } N\to\infty.
$$
\end{lemma}

These two lemmas together implies that as $N\to\infty$, the estimation error $f_N-f$ is asymptotically Gaussian with zero mean and diagonal covariance.
Similar asymptotic normality and uncorrelatedness properties hold for the derivatives of $f_N-f$ \cite{singh1981speed}.
In particular, as $N\to\infty$,
\begin{align}\label{eq:derivative asymptotic normality}
\begin{split}
    &\sqrt{Nh^{n+2}}(\partial_jf_N(x)-\operatorname{E}[\partial_jf_N(x)])\\
    &\qquad\qquad\to\mathcal{N}\Big(0,f(x)\int\big(\partial_jK(u)\big)^{2} du\Big)
\end{split}
\end{align}
in distribution.

\subsection{Generalized inverse operators}
\label{section:generalized inverse}

We introduce the notion of generalized inverse operators of self-adjoint operators defined in \cite{mou2006two}, which is based on an orthogonal decomposition.
Let $\mathcal{H}$ be a Hilbert space and $\Theta:\mathcal{D}(\Theta)\subseteq\mathcal{H}\to\mathcal{H}$ be a self-adjoint operator, where $\mathcal{D}(\Theta)$ represents the domain of $\Theta$.
Let $\mathcal{R}(\Theta)$ and $\mathcal{N}(\Theta)$ to be the range and kernel of $\Theta$, respectively.
Since $\Theta$ is self-adjoint, $\mathcal{N}(\Theta)^{\perp}=\overline{\mathcal{R}(\Theta)}$ (and we always have $\Theta(\mathcal{D}(\Theta)\cap\overline{\mathcal{R}(\Theta)})\subseteq\mathcal{R}(\Theta)$).
Thus, under the decomposition $\mathcal{H}=\mathcal{N}(\Theta)\oplus\overline{\mathcal{R}(\Theta)}$, we have the following representation for $\Theta$:
\begin{equation*}
\Theta=\begin{pmatrix}
0 & 0 \\
0 & \widehat{\Theta}
\end{pmatrix}
\end{equation*}
where $\widehat{\Theta}:\mathcal{D}(\Theta)\cap\overline{\mathcal{R}(\Theta)}\subseteq\overline{\mathcal{R}(\Theta)}\to\overline{\mathcal{R}(\Theta)}$ is self-adjoint.
Now, the generalized inverse $\Theta^{\dagger}$ is defined by:
\begin{equation*}
\Theta^{\dagger}=\begin{pmatrix}
0 & 0 \\
0 & \widehat{\Theta}^{-1}
\end{pmatrix}
\end{equation*}
with domain
\begin{align*}
\mathcal{D}(\Theta^{\dagger})&=\mathcal{N}(\Theta)+\mathcal{R}(\Theta)\\
&\equiv\{u^0+u^1\mid u^0\in\mathcal{N}(\Theta), u^1\in\mathcal{R}(\Theta)\}\supseteq\mathcal{R}(\Theta).
\end{align*}
We have the following facts:

(i) $\Theta^{\dagger}$ is self-adjoint.

(ii) One has that
\begin{equation*}
\Theta \Theta^{\dagger} \Theta=\Theta, \quad \Theta^{\dagger} \Theta \Theta^{\dagger}=\Theta^{\dagger}, \quad(\Theta^{\dagger})^{\dagger}=\Theta
\end{equation*}

(iii) 
The operator $\Theta\Theta^{\dagger}:\mathcal{D}(\Theta^{\dagger})\to\mathcal{H}$ is an orthogonal projection onto $\mathcal{R}(\Theta)$.
We may naturally extend it, still denoted it by itself, to $\overline{\mathcal{D}(\Theta^{\dagger})}=\mathcal{H}$.
Hence, $\Theta\Theta^{\dagger}:\mathcal{H}\to\overline{\mathcal{R}(\Theta)} \subseteq \mathcal{H}$ is the orthogonal projection onto $\overline{\mathcal{R}(\Theta)}$. Similarly, we can extend $\Theta^{\dagger} \Theta$ to be an orthogonal projection from $\mathcal{H}$ onto $\overline{\mathcal{R}(\Theta^{\dagger})}=\mathcal{N}(\Theta^{\dagger})^{\perp}=\mathcal{N}(\Theta)^{\perp}=\overline{\mathcal{R}(\Theta)}$. Therefore, we have
\begin{equation*}
\Theta \Theta^{\dagger}=\Theta^{\dagger} \Theta \equiv P_{\overline{\mathcal{R}(\Theta)}} \equiv \text { orthogonal projection onto } \overline{\mathcal{R}(\Theta)}.
\end{equation*}


\subsection{Dynamic average consensus}
\label{section:dynamic average consensus}

We introduce the PI consensus estimator presented in \cite{freeman2006stability}.
Consider a group of $N$ agents where each agent has a local reference signal $u_{i}(t):[0, \infty) \to \mathbb{R}$.
The dynamic average consensus problem consists of designing an algorithm such that each agent tracks the time-varying average  $u_{\text{avg}}(t):=\frac{1}{N}\sum_{i=1}^{N}u_i(t)$.
The PI consensus algorithm is given by \cite{freeman2006stability}:
\begin{align}
    \dot{\nu}_{i} &=-\alpha\left(\nu_{i}-u_{i}\right)-\sum_{j=1}^{N} a_{i j}\left(\nu_{i}-\nu_{j}\right)+\sum_{j=1}^{N} b_{j i}\left(\eta_{i}-\eta_{j}\right)\nonumber \\
    \dot{\eta}_{i} &=-\sum_{j=1}^{N} b_{i j}\left(\nu_{i}-\nu_{j}\right)\label{eq:PI consensus estimator}
\end{align}
where $u_{i}$ is agent $i$'s reference input, $\eta_{i}$ is an internal state, $\nu_{i}$ is agent $i$'s estimate of $u_{\text{avg}}$, $[a_{i j}]_{N\times N}$ and $[b_{i j}]_{N\times N}$ are adjacency matrices of the communication graph, and $\alpha>0$ is a parameter determining how much new information enters the dynamic averaging process. 
The Laplacian matrices associated with $\left[a_{i j}\right]$ and $\left[b_{i j}\right]$ are represented by $\mathbf{L}_\mathbf{P}$ (proportional) and $\mathbf{L}_\mathbf{I}$ (integral) respectively. 
The PI estimator solves the consensus problem under constant (or slowly-varying) inputs, and remains stable for varying inputs in the sense of ISS.
Define the tracking error of agent $i$ by $\epsilon_{i}(t)=\nu_{i}(t)-u_{\text{avg}}(t),i=1,\dots,N$.
Decompose the error into the consensus direction $\mathbf{1}_{N}$ and the disagreement directions orthogonal to $\mathbf{1}_{N}$. 
Define the transformation matrix $\mathbf{T}=[(1/\sqrt{N}) \mathbf{1}_{N} ~\mathbf{R}]$ where $\mathbf{R} \in \mathbb{R}^{N \times(N-1)}$ is such that $\mathbf{T}^{\top} \mathbf{T}=\mathbf{T} \mathbf{T}^{\top}=\mathbf{I}_{N}$. 
Consider the change of variables
$$
\bar{\boldsymbol{\epsilon}}=\left[\begin{array}{c}
\bar{\epsilon}_{1} \\
\bar{\boldsymbol{\epsilon}}_{2: N}
\end{array}\right]=\mathbf{T}^{\top}\boldsymbol{\epsilon}, 
\quad 
\mathbf{w}=\left[\begin{array}{c}
w_{1} \\
\mathbf{w}_{2:N}
\end{array}\right]=\mathbf{T}^{\top} \mathbf{\eta},
$$
\begin{equation*}
    \mathbf{y}=\mathbf{w}_{2:N}-\alpha(\mathbf{R}^{\top}\mathbf{L}_{\mathbf{I}}^{\top}\mathbf{R})^{-1}\mathbf{R}^{\top}\dot{\mathbf{u}}
\end{equation*}
to write \eqref{eq:PI consensus estimator} in the equivalent form
\begin{align*}
    \begin{split}
        \dot{w}_{1}=& 0, \\
        \begin{bmatrix}
            \dot{\mathbf{y}} \\ 
            \dot{\bar{\epsilon}}_{1} \\ 
            \dot{\bar{\boldsymbol{\epsilon}}}_{2: N}
        \end{bmatrix}
        =&\underbrace{\begin{bmatrix}
        \mathbf{0} & \mathbf{0} & -\mathbf{R}^{\top} \mathbf{L}_{\mathbf{I}} \mathbf{R} \\ 
        0 & -\alpha & 0 \\ 
        \mathbf{R}^{\top} \mathbf{L}_\mathbf{I}^{\top} \mathbf{R} & 0 & -\alpha \mathbf{I}-\mathbf{R}^{\top} \mathbf{L}_\mathbf{P} \mathbf{R}
        \end{bmatrix}}_{\mathbf{A}}
        \begin{bmatrix}
        \mathbf{y} \\
        \bar{\epsilon}_{1} \\
        \bar{\boldsymbol{\epsilon}}_{2:N}
        \end{bmatrix}\\ 
        &+\underbrace{\begin{bmatrix}
        -\alpha(\mathbf{R}^{\top}\mathbf{L}_\mathbf{I}^{\top}\mathbf{R})^{-1} \\
        0 \\
        0
        \end{bmatrix}}_{\mathbf{B}}\mathbf{R}^{\top}\dot{\mathbf{u}}.
    \end{split} 
\end{align*}
The stability result is given as follows.

\begin{lemma}\label{lmm:ISS of PI estimator}
\cite{kia2019tutorial} Let $\mathbf{L}_\mathbf{I}$ and $\mathbf{L}_\mathbf{P}$ be Laplacian matrices of strongly connected and weight-balanced digraphs. 
Then
\begin{align*}
\begin{split}
    |\epsilon_{i}(t)|\leq & \kappa e^{-\lambda(t-t_{0})}\left\|\begin{bmatrix}
    \mathbf{w}_{2: N}(t_{0}) \\
    \overline{\boldsymbol{\epsilon}}(t_{0})
    \end{bmatrix}
    \right\| \\
    &+\frac{\kappa\|\mathbf{B}\|}{\lambda}\sup_{t_{0}\leq\tau\leq t}\left\|\left(\mathbf{I}_{N}-\frac{1}{N} \mathbf{1}_{N} \mathbf{1}_{N}^{\top}\right) \dot{\mathbf{u}}(\tau)\right\|
\end{split}
\end{align*}
where $\kappa,\lambda>0$ are constants depending on $\mathbf{A}$.
\end{lemma}

For switching networks, each switch introduces a transient to the estimator error, and $\epsilon_i$ is still ISS. 
The PI consensus estimator is also robust to initialization errors and permanent agent dropout in the sense that \cite{kia2019tutorial}:
\begin{equation}\label{eq:robust to initialization}
    \sum_{i=1}^{N} \nu_{i}(t)=\sum_{i=1}^{N} u_{i}(t)+e^{-\alpha\left(t-t_{0}\right)}\left(\sum_{i=1}^{N} \nu_{i}(t_{0})-\sum_{i=1}^{N} u_{i}(t_{0})\right).
\end{equation}

\bibliographystyle{IEEEtran}
\bibliography{References}

\end{document}